\documentclass{article}

    \usepackage[preprint]{neurips_2024}

\usepackage[utf8]{inputenc} % allow utf-8 input
\usepackage[T1]{fontenc}    % use 8-bit T1 fonts
\usepackage[colorlinks,linkcolor=black,citecolor=black,urlcolor=black,bookmarks=true,hypertexnames=false]{hyperref}       % hyperlinks
\usepackage{url}            % simple URL typesetting
\usepackage{booktabs}       % professional-quality tables
\usepackage{amsfonts}       % blackboard math symbols
\usepackage{nicefrac}       % compact symbols for 1/2, etc.
\usepackage{microtype}      % microtypography
\usepackage{xcolor}         % colors
\usepackage[small,bf]{caption}
\usepackage{graphicx,amsthm}
\usepackage[linesnumbered,ruled]{algorithm2e}

\newcommand{\EE}{\mathbb{E}}

\newcommand{\argmin}{\mathop{\rm argmin}}

\newtheorem{thm}{Theorem}[section]
\newtheorem{lem}{Lemma}[section]

\newtheorem{prop}{Proposition}[section]
\newtheorem{asmp}{Assumption}[section]
\newtheorem{defn}{Definition}[section]

\newtheorem{rem}{Remark}[section]

\def\approxcorrect{\checkmark\kern-1.1ex\raisebox{.89ex}{$\times$}}

\usepackage{amsmath,amsfonts,bm}

\def\eqref#1{equation~\ref{#1}}
\def\1{\bm{1}}

\DeclareMathAlphabet{\mathsfit}{\encodingdefault}{\sfdefault}{m}{sl}
\SetMathAlphabet{\mathsfit}{bold}{\encodingdefault}{\sfdefault}{bx}{n}

\newcommand{\KL}{D_{\mathrm{KL}}}

\SetKwInput{KwInput}{Input}                % Set the Input
\SetKwInput{KwOutput}{Output}              % set the Output
\title{Last-iterate Convergence in Regularized Graphon Mean Field Game}

\author{%
Jing Dong \\
  The Chinese University of Hong Kong, Shenzhen\\
  \texttt{jingdong@link.cuhk.edu.cn} \\
   \And
   Baoxiang Wang\\
   The Chinese University of Hong Kong, Shenzhen \\
   \texttt{bxiangwang@cuhk.edu.cn} \\
   \AND
   Yaoliang Yu \\
   University of Waterloo \\
   \texttt{yaoliang.yu@uwaterloo.ca} \\
}

\begin{document}

\maketitle

\begin{abstract}
To model complex real-world systems, such as traders in stock markets, or the dissemination of contagious diseases, graphon mean-field games (GMFG) have been proposed to model many agents. 
%and to use the Graphons for heterogeneous algents. 
Despite the empirical success, our understanding of GMFG is limited. Popular algorithms such as mirror descent are deployed but remain unknown for their convergence properties. In this work, we give the first last-iterate convergence rate of mirror descent in regularized monotone GMFG. In tabular monotone GMFG with finite state and action spaces and under bandit feedback, we show a last-iterate convergence rate of $O(T^{-1/4})$. Moreover, when exact knowledge of costs and transitions is available, we improve this convergence rate to $O(T^{-1})$, matching the existing convergence rate observed in strongly convex games.
%(though not specifically mean field). 
%For scenarios involving large or even infinite state spaces, we extend our analysis to linear GMFGs, characterized by costs and transitions adhering to a linear structure. In this context, 
In linear GMFG, our algorithm achieves a last-iterate convergence rate of $O(T^{-1/5})$. Finally, we verify the performance of the studied algorithms by empirically testing them against fictitious play in a variety of tasks. 
\end{abstract}
\section{Introduction}

In many real-world applications, the presence of complex systems including many interacting individuals or components is indispensable. These systems manifest in various forms, from the intricate networks of neurons within human brains \citep{bullmore2009complex,bullmore2012economy,avena2018communication}, to the dynamic interactions of traders in stock markets \citep{bakker2010social,bian2016evolving}, and to the dissemination of contagious diseases throughout societies \citep{newman2002spread,pastor2015epidemic}. Due to the large number of interacting individuals or components, these systems pose significant challenges for modeling. Mean field games (MFG) \citep{caines2006large,lasry2007mean} have emerged as a highly effective approach for addressing this complexity, offering both scalability and robust theoretical guarantees in these multi-agent systems. MFG operates on the principle of weak interactions, positing that each individual's influence on the overall system is negligible. This framework has been successfully applied to many real-world tasks, including social networks \citep{yang2018learning}, and swarm robotics \citep{cui2023scalable}.

The MFG framework leverages the assumption of agent homogeneity and has demonstrated success across various applications. However, this assumption becomes a hindrance when dealing with heterogeneous agents. To address this limitation, the Graphon mean field games (GMFG) \cite{parise2019graphon,aurell2022finite} framework has been introduced as an extension of MFGs to accommodate heterogeneous agent modeling. The GMFG framework captures agent interactions through a graphical structure and is shown to be successful in applications like modeling investment decisions in financial markets \citep{tangpi2024optimal}. Despite the empirical success of MFG and GMFG, our theoretical understanding of this framework remains limited. 

In monotone MFGs \citep{lasry2007mean}, and under the access to the exact cost and transition functions, \cite{perrin2020fictitious} proposed a continuous time fictitious play algorithm, where the averaged iterates policy converge to a Nash equilibrium in $O(T^{-1})$ iterations. For discrete-time monotone GMFGs with access to exact cost and transition functions,  \cite{zhang2023learning} proposed a mirror descent-based algorithm that converges to the Nash equilibrium in $O(T^{-1/2})$ iterations. 

However, in real-world applications, approximating continuous-time dynamics can be challenging, and exact knowledge of cost and transition functions may not be feasible. Agents typically only receive bandit feedback, meaning they observe the cost and transitions associated with the states and actions they have visited. Meanwhile, the existing approaches only guarantee the convergence of the time average of the joint action profile, rather than the last-iterate convergence, the convergence of the joint action profile. Last-iterate convergence holds greater appeal as it offers a descriptive account of the evolution of players' overall behavior. In contrast, while the trajectory of players' joint action converge in the time-average sense, it may exhibit cycling, which is not suitable for practical deployment \citep{mertikopoulos2018cycles}. The following question thus arises.
\begin{center}
\textit{How fast can discrete-time algorithms converge (in the last iterate) to a Nash equilibrium in GMFGs with bandit feedback?}
\end{center}

In this work, we focus on the mirror descent-based algorithm, which has been empirically verified to be successful in GMFGs \citep{perolat2022scaling,zhang2023learning}. In tabular monotone GMFGs (finite state and actions space) and under bandit feedback, we show a $O(T^{-1/4})$ last-iterate convergence rate. When the exact knowledge of cost and transitions is present, we show that the convergence rate can be improved to $O(T^{-1})$, matching the existing convergence rate in strongly convex (but not mean field) games. To address scenarios involving large or even infinite state spaces, we extend our analysis to linear GMFGs, where costs and transitions adhere to a linear structure. In this context, we achieve a last-iterate convergence rate of $O(T^{-1/5})$. We validate the effectiveness of the studied algorithm by empirically comparing them against the fictitious play in four different environments. 

\section{Related Works}
\paragraph{Mean Field Game (MFG)}
To address the challenge of modeling a large number of agents in a game, the Mean Field Game (MFG) was proposed by \cite{caines2006large,lasry2007mean}. It considers the limit case of a continuous distribution of homogeneous agents (all anonymous and with symmetric interest) and reduces the problem to the characterization of the optimal behavior of a single representative agent. The classic approaches include the numerical approximation approach for partial differential equation \citep{achdou2010mean,achdou2012mean,achdou2020mean}, and the more recent deep reinforcement learning approaches \citep{cui2021approximately,lauriere2022scalable,fabian2023learning}. 

Recent efforts also introduced the traditional fictitious play (FP) algorithm and combined it with machine learning techniques \citep{perrin2020fictitious}. While FP achieves impressive results and is shown to be convergent \citep{geist2022concave}, it is hard to scale due to its low computational efficiency, as it requires computing the best response at every iteration. To address this, the policy mirror descent algorithm is proposed, and its asymptotic convergence in continuous time is studied \cite{perolat2022scaling}. Under a monotonicity assumption and regularization, the average iterate of the discrete-time mirror descent algorithm is shown to enjoy linear convergence under the tabular case and with function approximation (but with access to an approximation subroutine) \citep{zhang2023learning}.

\paragraph{Graphon Mean Field Game (GMFG)}
To capture the heterogeneity among agents, the Graphon mean field game (GMFG) has been proposed by \cite{parise2019graphon}, where the heterogeneous interaction between agents is described by graphon. By using a contraction condition, \cite{cui2021learning} proposed an algorithm to efficiently approximate the Nash equilibrium. The average iterate of the mirror descent algorithm is then shown to be convergent \citep{fabian2023learning} (asymptotically) and in finite time \citep{zhang2023learning}. To our best knowledge, there is no last-iterate convergence guarantee of algorithms for (graphon) mean field games.

\paragraph{Last-iterate Convergence in Monotone (not Mean Field) Game}
In strongly convex games with full gradient feedback, the linear last-iterate convergence rate is established \citep{tseng1995linear,liang2019interaction,zhou2020convergence}. When the gradient feedback is with a zero-mean noise, \cite{jordan2022adaptive} gave a $O(T^{-1})$ last-iterate convergence rate. When only bandit feedback is available, \cite{bervoets2020learning} established an asymptotic convergence rate if the equilibrium is unique. Subsequently, \cite{bravo2018bandit} improved this convergence rate of $O(T^{-1/3})$, while the proposed algorithm also ensured the no-regret property. Later works by \cite{lin2021doubly} further improved the last-iterate convergence rate to $O(T^{-1/2})$ using the self-concordant barrier function.

\section{Preliminary}
We consider a GMFG defined as $\left( \mathcal{I}, \mathcal{S}, \mathcal{A}, \{P\}_{h \in [H]}, \{c\}_{h \in [H]}, \{W_h\}_{h \in [H]}, \mu_1, H\right)$ with infinitely many agents. Each agent corresponds to a point $\alpha \in \mathcal{I}$. Let $\nu$ be a positive measure on $\mathcal{I}$. The state and action space ($\mathcal{S}$ and $\mathcal{A}$) are the same for each agent. We further assume the state space is compact and the action space is finite. The interaction among agents at time $h$ is captured through graphon $W_h$, a symmetric function such that $W_h(\alpha, \beta) = W_h(\beta, \alpha)$. The transition and reward of each agent are affected by the collective behavior of all other agents by an aggregate $z$. At time $h$, the aggregate for agent $\alpha$ is defined as $z_h^\alpha = \int_\mathcal{I} W_h(\alpha, \beta) \mu_h^\beta d\nu(\beta)$, where $\mu_h^\beta$ is the state distribution of agent $\beta$. We assume each agent has access to the aggregate $z_h^\alpha$. We also let $\mu^\mathcal{I}(s) = \lim_{N\rightarrow\infty} \frac{\sum^N_{j=1}\mathbb{I}\{s_j = s\}}{N}$ denote the state distribution of all agents. On state $s_h^\alpha$ and when the agent takes action $a_h^\alpha$, the state transits according to $s_{h+1}^\alpha \sim P_h(\cdot \mid s_h^\alpha, a_h^\alpha, z_h^\alpha)$. The agent $\alpha$ will also incur a cost of $c_h(s_h^\alpha, a_h^\alpha, z_h^\alpha)$. 

Define the value functions as
\begin{align}\label{eq:VandQ}
    V^{\alpha}_{h}\left(s^\alpha, \pi^\alpha, \mu^{\mathcal{I}}\right) = \ & \EE_{\pi^\alpha} \left[\sum^H_{t=h} c_h\left(s_t^\alpha, a_t^\alpha, z_h^\alpha\left( \mu^{\mathcal{I}}\right)\right) \mid s_{h}^\alpha = s^\alpha\right] \,, \\
    Q^{\alpha}_{h}\left(s^\alpha, a^\alpha, \pi^\alpha, \mu^{\mathcal{I}}\right) = \ &  c_h\left(s^\alpha,a^\alpha, z_{h}^\alpha\right) + \EE_{\pi^\alpha, P_h} \left[V^{\alpha}_{h + 1} \left(s_{h + 1}^\alpha, \pi^\alpha, \mu^{\mathcal{I}}\right)\mid s_{h}^\alpha = s^\alpha, a_h^\alpha = a^\alpha\right] \,.
\end{align}
Following the standard settings studied in GMFG and Markov games, we investigate the convergence with the regularized value function, which enables faster convergence \citep{zhang2023learning,cen2021fast,shani2020adaptive}.
The $\lambda$-regularized value functions are defined as 
\begin{align}\label{eq:lambda_VandQ}
    V^{\lambda, \alpha}_{h}\left(s^\alpha, \pi^\alpha, \mu^{\mathcal{I}}\right) = \ & \EE_{\pi^\alpha} \left[\sum^H_{t=h} c_h\left(s_t^\alpha, a_t^\alpha, z_t^\alpha\left( \mu^{\mathcal{I}}\right)\right) + \lambda \ln \pi_t^\alpha(a_t^\alpha \mid s_t^\alpha) \mid s_{h}^\alpha = s^\alpha\right] \,, \\
    Q^{\lambda, \alpha}_{h}\left(s^\alpha, a^\alpha, \pi^\alpha, \mu^{\mathcal{I}}\right) = \ &  c_h\left(s^\alpha,a^\alpha, z_{h}^\alpha\right) + \EE_{\pi^\alpha, P_h} \left[V^{\lambda, \alpha}_{h + 1} \left(s_{h + 1}^\alpha, \pi^\alpha, \mu^{\mathcal{I}}\right)\mid s_{h}^\alpha = s^\alpha, a_h^\alpha = a^\alpha\right] \,.
\end{align}
Without loss of generality, we further assume the rewards are bounded between $[0,1]$. Then, $\left\|Q_h^{\lambda, \alpha}\left(s^\alpha, \cdot, \pi_t^\alpha, \mu_t^\mathcal{I}\right)\right\|_\infty \leq H$, for any $h, \alpha, s$. We define cumulative reward as 
\begin{align*}
    J^{\alpha}(\pi, \mu_1)=\mathbb{E}_{\mu_1}\left[V_1(s^\alpha, \pi, \mu)\right]\,, \quad J^{\lambda,\alpha}(\pi, \mu_1)=\mathbb{E}_{\mu_1}\left[V_1^\lambda(s^\alpha, \pi, \mu)\right] \,.
\end{align*}

A common solution concept in GMFG is Nash equilibrium, which equilibrium state where no agent can gain in value by unilaterally changing its action. Formally, the Nash equilibrium is defined as follows. 
\begin{defn}[Nash equilibrium]
    An NE of the $\lambda$-regularized MFG is a pair $\left(\pi^{\ast, \mathcal{I}}, \mu^{\ast, \mathcal{I}}\right)$ that satisfies
    \begin{itemize}
        \item Agent rationality: $J^{\lambda, \alpha}\left(\pi^{\ast, \alpha}, \mu^{\ast, \mathcal{I}}\right)=\min _{\widetilde{\pi}^\alpha \in \Pi^H} J^{\lambda, \alpha}\left(\widetilde{\pi}^\alpha, \mu^{\ast, \mathcal{I}}\right)$ for all $\alpha \in \mathcal{I}$ up to a zero measure set on $\mathcal{I}$ with respect to $\nu$. 
        \item Distribution consistency: The distribution flow $\mu^{*, \mathcal{I}}$ is equal to the distribution flow induced by implementing the policy $\pi^{*, \mathcal{I}}$.
    \end{itemize}
\end{defn}

To ensure the existence of a Nash equilibrium in a $\lambda$-regularized GMFG, we maintain the following assumptions of the game.
\begin{asmp}\label{asmp:continuous}
    The GMFG satisfies
    \begin{itemize}
        \item The cost function and the transition function are continuous. 
        \item The graphon is a continuous function.
    \end{itemize}
\end{asmp}

\begin{rem}
The above model also includes games with finitely many players.
For a finite graph, $\mathcal{G} = (V, E)$ with $N$ nodes denoting the agents, and $E$ denotes the set of edges that models the relationship between agents. We can partition a unit interval $[0,1]$ into $N$ intervals, $I_1, \ldots, I_N$ of equal length. Then we can let the graphon $W$ assign a constant value on each square $I_i \times I_j$, $i, j \in V$. It is equal to one if there is an edge between $i, j$ in $\mathcal{G}$, and zero otherwise. Although this is not continuous, one can smooth it so that it is continuous \citep{fabian2023learning}. 
\end{rem}

\begin{thm}[Theorem 4.4 \cite{zhang2023learning}]
    Under Assumption \ref{asmp:continuous}, for all $\lambda \geq 0$, there exists a Nash equilibrium in a $\lambda$-regulairzed GMFG. 
\end{thm}

To ensure the uniqueness of Nash equilibrium, we further maintain the following weakly monotonicity assumption of the game. The following condition is a generalization of the monotonicity condition in games \citep{lin2020finite,lin2021doubly,duvocelle2023multiagent}, and is commonly seen in literature in GMFG \citep{zhang2023learning}. 

\begin{asmp}[Weakly monotone condition]\label{asmp:monotone}
A GMFG is said to be weakly monotone if for any $\rho^{\mathcal{I}}, \tilde{\rho}^{\mathcal{I}} \in \Delta(\mathcal{S} \times \mathcal{A})^{\mathcal{I}}$ and their marginalizations on the states $\mu^{\mathcal{I}}, \tilde{\mu}^{\mathcal{I}} \in \Delta(\mathcal{S})^{\mathcal{I}}$, we have
$
\int_{\mathcal{I}} \sum_{a \in \mathcal{A}} \int_{\mathcal{S}}\left(\rho^\alpha(s, a)-\widetilde{\rho}^\alpha(s, a)\right)\left(c_h\left(s, a, z_h^\alpha\left(\mu^{\mathcal{I}}\right)\right)-c_h\left(s, a, z_h^\alpha\left(\widetilde{\mu}^{\mathcal{I}}\right)\right)\right) \mathrm{d} s \mathrm{~d} \nu(\alpha) \geq 0 
$,
for all $t$. It is strictly weakly monotone if the inequality is strict when $\rho^{\mathcal{I}} \neq \widetilde{\rho}^{\mathcal{I}}$.
\end{asmp}

When a $\lambda$-regularized GMFG admits Assumption \ref{asmp:monotone}, there exists a unique Nash equilibrium and satisfies the following property. An example of such a weakly monotone game is the multi-population predator-prey model described in \citep{perolat2022scaling}.

\begin{prop}[Proposition 5.3 of \cite{zhang2023learning}]\label{prop:monotone}
If a $\lambda$-regularized GMFG satisfies the weakly monotone condition, then for any two policies $\pi^{\mathcal{I}}, \widetilde{\pi}^{\mathcal{I}} \in \widetilde{\Pi}$ and their induced distribution flows $\mu^{\mathcal{I}}, \widetilde{\mu}^{\mathcal{I}} \in \widetilde{\Delta}$, we have
$
\int_\mathcal{I} J^{\lambda, \alpha}\left(\pi^\alpha, \mu^{\mathcal{I}}\right)+J^{\lambda, \alpha}\left(\widetilde{\pi}^\alpha, \tilde{\mu}^{\mathcal{I}}\right)-J^{\lambda, \alpha}\left(\widetilde{\pi}^\alpha, \mu^{\mathcal{I}}\right)-J^{\lambda, \alpha}\left(\pi^\alpha, \widetilde{\mu}^{\mathcal{I}}\right) \mathrm{d} \nu(\alpha) \geq 0 
$.

If the $\lambda$-regularized GMFG satisfies the strictly weakly monotone condition, then the inequality is strict when $\pi^{\mathcal{I}} \neq \widetilde{\pi}^{\mathcal{I}}$.
\end{prop}

\section{Algorithm}
In this section, we introduce our algorithm for solving Nash equilibrium in regularized GMFG. Our algorithm extends the celebrated mirror-descent algorithm, for which its efficiency in solving Nash equilibrium has been demonstrated. The last-iterate convergence properties of mirror descent has been investigated in many works \citep{cen2021fast,lin2021doubly,cai2023uncoupled,duvocelle2023multiagent}. 

At each iteration $t$, the agent $\alpha$ execute $\{\pi_{t,h}^\alpha\}^H_{h=1}$ for $H$ steps and receive costs $\{c_h(s_h^\alpha, a_h^\alpha, z_h^\alpha)\}^H_{h=1}$. Then, dependent on the information available, the agent computes a gradient $\hat{g}_{t,h}(s^\alpha, \cdot)$ and updates it with a mirror descent step. Algorithm \ref{alg} provides a summary of our algorithm. 

\begin{algorithm}
    \KwInput{Learning rate $\{\eta_t\}^T_{t=1}$, regularization constant $\lambda$}
    \For{$t = 1, \ldots, T$}{
    \For{$h = H, \ldots, 1$}{
    Execute $\pi_{t,h}^\alpha$ and receive costs $c_h(s_h^\alpha, a_h^\alpha, z_h^\alpha)$\;
    }
    \For{$h = H, \ldots, 1$}{
    Compute gradient $\hat{g}_{t,h}(s^\alpha_h, \cdot)$ according to Equation (\ref{eq:tabular_approximation}), or (\ref{eq:linear_approx_Q})\;
    
    $\pi_{t+1, h}^\alpha(\cdot \mid s_h^\alpha) = \arg\min_{\pi^\alpha} \eta_t \left\langle  \hat{g}_{t,h}(s_h^\alpha, \cdot) +\lambda\log(\pi_t^\alpha), \pi^\alpha(\cdot \mid s_h^\alpha)\right\rangle + \KL\left(\pi^\alpha (\cdot \mid s_h^\alpha), \pi_{t,h}^\alpha(\cdot \mid s_h^\alpha)\right)$\;
    
    }
    }
    \caption{Tabular online mirror descent for $\lambda$-regularized GMFG}\label{alg}
\end{algorithm}
We now discuss how one can construct the gradient estimator in a full information setting, tabular bandit feedback setting, and linear GMFG setting. 
\paragraph{Tabular GMFG with full information feedback}
When the cost function and the transition kernel are known to the agent, the agent can set $\hat{g}_{t,h}(s^\alpha, \cdot) = Q^{\lambda, \alpha}_{t,h}\left(s^\alpha,\cdot, \pi^\alpha_{t,h}, \mu^{\mathcal{I}}\right)$, which can be computed via value iteration,  
\begin{align*}
    Q^{\lambda, \alpha}_{t,h}\left(s^\alpha,\cdot, \pi^\alpha_{t}, \mu^{\mathcal{I}}\right)
    = \ & c_h\left(s^\alpha, \cdot, z_{h}^\alpha\right) + P_h\left(s^\alpha, \cdot, z_h^\alpha\right) V^{\lambda, \alpha}_{h + 1}\,,\\
    V^{\lambda, \alpha}_{t, h} \left(s^\alpha, \pi^\alpha_{t}, \mu^{\mathcal{I}}\right) = \ & \left\langle Q^{\lambda, \alpha}_{h}\left(s^\alpha,\cdot, \pi^\alpha_{t,h}, \mu^{\mathcal{I}}\right),  \pi_{t,h}^\alpha(\cdot \mid s^\alpha) \right\rangle \,.
\end{align*}

\paragraph{Tabular GMFG with bandit feedback}
Under the tabular bandit feedback model, the agent does not have exact knowledge of the cost and transition kernel. Instead, they can only observe the cost corresponding to the state, action, and state distribution that they have visited. In this case, we compute the gradient as follows, which is similar to the gradient estimator on multi-agent tabular Markov games \citep{jin2022v}. Let $k = N_{t,h}(s)$ be the number of times state $s$ is visited at step $h$ up to time t. Then we first approximate the value function as 
\begin{align}
    \hat{V}^{\lambda,\alpha}_{h}\left(s_h^\alpha, \pi_t^\alpha, \mu^\mathcal{I}\right) 
    = \ &  (1-\beta_k)\hat{V}^{\lambda,\alpha}_{h}\left(s_h^\alpha,\pi_t^\alpha, \mu^\mathcal{I}\right) + \beta_k \left(c_h(s_h^\alpha, a_h^\alpha, z_h^\alpha) + \hat{V}^{\lambda,\alpha}_{t,h+1}\left(s_{h+1}^\alpha, \pi_t^\alpha, \mu^\mathcal{I}\right)\right) \\
    \hat{V}^{\lambda,\alpha}_{h}\left(s_h^\alpha, \pi_t^\alpha, \mu^\mathcal{I}\right)
    = \ & \max\left\{H+1-h, \hat{V}^{\lambda,\alpha}_{h}\left(s_h^\alpha, \pi_t^\alpha, \mu^\mathcal{I}\right)\right\} \,.
\end{align}
As we have no access to the cost associated with each action, we estimate the gradient with respect to all actions using importance sampling
\begin{align}\label{eq:tabular_approximation}
    \hat{g}_{t,h}(s_h^\alpha, a^\alpha) = \ & \frac{\mathbb{I}\{a_h^\alpha = a^\alpha\}\left(c_h + \hat{V}^{\lambda,\alpha}_{t,h+1}\left(s_{h+1}^\alpha, \pi_t^\alpha, \mu^\mathcal{I}\right)\right)}{\pi_{t,h}^\alpha(a^\alpha \mid s_h^\alpha) + \gamma_t} \,.
\end{align}
To avoid unbounded gradient estimation and to encourage exploration, we add a $\gamma_t$ factor. For the initialization step, we set $\hat{V}(s, \pi_0^\alpha, \mu^\mathcal{I}) = 0$, for all $s$.

\paragraph{Linear GMFG with bandit feedback}
When the state and action space are large, maintaining a tabular value function may become infeasible. Therefore, we consider learning within a linearly parameterized GMFG, which extends the tabular GMFG and accommodates the potential enormity of state and action spaces.
\begin{defn}[Linear GMFG]
    A linear GMFG has a linearly structured transition $P_h(\cdot \mid s_h^\alpha, a_h^\alpha, z_h^\alpha) = \theta_h^\ast \phi(s_h^\alpha, a_h^\alpha, z_h^\alpha)\,,\quad \forall h, s_h^\alpha \in \mathcal{S}, a_h^\alpha\in \mathcal{S}$, where $\phi$ is a known feature mapping. Further, we assume 
    \begin{enumerate}
        \item $\sup_{s,a,z}\|\phi(s,a,z)\|_2 \leq 1$, and
        \item $\|v^\top \theta_h^\ast \| \leq \sqrt{d}$, for any $\|v\|_\infty \leq 1$ and all $h$.
    \end{enumerate}
\end{defn}
Given the linear structure of the transition kernel, it remains to estimate $\theta_h$ accurately to compute the value function. Let $\delta_h(s^\alpha)$ be a one-hot vector that has zero everywhere except that the entry corresponding to $s^\alpha$ is one, and denote $\epsilon_h^\alpha = P_h(\cdot \mid s_h^\alpha, a_h^\alpha, z_h^\alpha) - \delta_h(s_{h+1}^\alpha)$. Conditioned on the history generated on all previous episodes up to episode $t$, $\mathcal{H}_{t,h}$, we have $\EE[\epsilon_h^\alpha \mid \mathcal{H}_{t,h}] = 0$. Therefore $\delta$ acts as an unbiased estimate of $P_h$. 

At iteration $t$, step $h$, we consider using all previous interactions $D_{t,h} = \{s_{j,h}^\alpha, a_{j,h}^\alpha, z_{j,h}^\alpha\}^{t-1}_{j=1}$to estimate the $\theta_h^\ast$. Once we have an estimate $\hat{\theta}_h$, we can use value iteration to compute the value function. With the dataset $D_{t,h}$ one could estimate $P_h$ using ridge regression
\begin{align*}
    \hat{\theta}_{t,h} = \argmin_{\theta_h} \sum^{t-1}_{j=1}\left\|\theta_h \phi\left(s_{j,h}^\alpha, a_{j,h}^\alpha, z_{j,h}^\alpha\right) - \delta_h \left(s_{j,h+1}^\alpha\right)\right\| + \|\theta_h\|^2 \,,
\end{align*}
for which the closed-form solution is
\begin{align}\label{eq:close_form}
    \hat{\theta}_{t,h} = \ & \sum^{t-1}_{j=1}\delta_h \left(s_{j,h+1}^\alpha\right)\phi\left(s_{j,h}^\alpha, a_{j,h}^\alpha, z_{j,h}^\alpha\right)^\top \left(\Lambda_t\right)^{-1}\,, \\
    \Lambda_t = \ & \sum^{t-1}_{j=1} \phi\left(s_{j,h}^\alpha, a_{j,h}^\alpha, z_{j,h}^\alpha\right)\phi\left(s_{j,h}^\alpha, a_{j,h}^\alpha, z_{j,h}^\alpha\right)^\top + I \,.
\end{align}
Using the estimate, we then update the value function using value iteration.
\begin{align}
    \hat{Q}^{\lambda, \alpha}_{t,h}\left(s^\alpha_h, a^\alpha_h, \pi^\alpha_t, \mu^{\mathcal{I}}\right) = \ &  c_h\left(s^\alpha_h,a^\alpha_h, z_{h}^\alpha\right) + \phi(s^\alpha, a^\alpha)^\top \hat{\theta}_{t,h}^\top V^{\lambda, \alpha}_{h + 1} \left(s_{h + 1}^\alpha, \pi^\alpha_t, \mu^{\mathcal{I}}\right)\,.
\end{align}
Similar to the tabular case, as we only have information on the state and action we have visited, we estimate the gradient using importance sampling,
\begin{align}\label{eq:linear_approx_Q}
    \hat{g}_{t,h}(s_h^\alpha, a^\alpha) = \ &  \frac{\mathbb{I}\{a_h^\alpha = a^\alpha\}\hat{Q}^{\lambda, \alpha}_{t,h}\left(s^\alpha_h, a^\alpha_h, \pi^\alpha_t, \mu^{\mathcal{I}}\right) }{\pi_{t,h}^\alpha(s_h^\alpha) + \gamma_t} \,.
\end{align}

\section{Convergence analysis}
In this section, we present our main results on the last-iterate convergence to the Nash equilibrium in a regularized GMFG. To measure the distance to the equilibrium, we use the following convergence metric.  
\paragraph{Convergence metric}

Define 
\begin{align*}
    D\left(\pi^\mathcal{I}_t\right) = \int_\mathcal{I} \sum^H_{h=1} \EE_{\mu_h^{\ast, \alpha}} \left[\KL\left( \pi^{\ast, \alpha}(\cdot \mid s_h^\alpha) , \pi_t^{\alpha}(\cdot \mid s_h^\alpha)\right)\right] d\nu(\alpha) \,.
\end{align*}
Note that this measures the weighted KL divergence between the policy computed and the Nash equilibrium, where the weights are the Nash equilibrium distribution flow $\mu^\ast$. At equilibrium, the metric is zero. We also note that this is used in \cite{zhang2023learning}.
\paragraph{Tabular GMFG with full information feedback}
When we have access to the exact cost and transition function, Theorem \ref{thm:exact_Q} shows we can converge linearly. 
\begin{thm}\label{thm:exact_Q}
Let $\eta_t = t^{-1}$. We have $D\left(\pi^\mathcal{I}_{t+1}\right)
    \leq  \frac{H^3}{\lambda t}$.
\end{thm}
We first fix a tuple $h, \alpha, s_h^\alpha$. To obtain Theorem \ref{thm:exact_Q}, we first use the mirror descent update rule to obtain the following relationship on $\KL\left( \pi^{\ast, \alpha}(\cdot \mid s_h^\alpha), \pi_{t+1,h}^{\alpha}(\cdot \mid s_h^\alpha)\right) $, 
\begin{align*}
    & \KL\left( \pi^{\ast, \alpha}(\cdot \mid s_h^\alpha) , \pi_{t+1,h}^{\alpha}(\cdot \mid s_h^\alpha)\right) \\
     \leq \ & \KL\left( \pi^{\ast, \alpha}(\cdot \mid s_h^\alpha) , \pi_{t,h}^{\alpha}(\cdot \mid s_h^\alpha)\right) + \frac{\eta_t^2H^2}{2} + \eta_t \left\langle \pi^{\ast, \alpha}(\cdot \mid s_h^\alpha) - \pi_{t,h}^{\alpha}(\cdot \mid s_h^\alpha)  , \right. \\
     & \ \left.Q^{\lambda, \alpha}_{h}\left(s_h^\alpha,\cdot, \pi^\beta_t, \mu^{\mathcal{I}}\right) +\lambda\log(\pi_{t,h}^{\alpha}(\cdot \mid s_h^\alpha)) \right\rangle  \,.
\end{align*}
We then use the third term and the monotonicity condition to obtain a recursion on $\KL\left( \pi^{\ast, \alpha}(\cdot \mid s_h^\alpha) , \pi_{t+1,h}^{\alpha}(\cdot \mid s_h^\alpha)\right) $. The observation is that the $\lambda$ parameter acts as a regularization for making the game more convex. Under $\lambda$ regularization, notice that $Q^{\lambda, \alpha}_{h}\left(s_{h}^\alpha,\cdot, \pi^{\alpha} , \mu^{\mathcal{I}}_t \right) + \lambda\log(\pi^{\alpha}(\cdot \mid s_h^\alpha))$ is the gradient. One can then show that 
\begin{align*}
    & \EE_{\pi^{\ast, \alpha}, \mu^{\mathcal{I}}_t }\left[\sum^H_{h=1}  \left\langle Q^{\lambda, \alpha}_{h}\left(s_{h}^\alpha,\cdot, \pi^{\alpha} , \mu^{\mathcal{I}}_t \right) + \lambda\log(\pi^{\alpha}(\cdot \mid s_h^\alpha)), \pi_{h}^{\ast,\alpha}(\cdot \mid s_h^\alpha) - \pi_{h}^{\alpha}(\cdot \mid s_h^\alpha)\right\rangle  \mid s_1^\alpha = s^\alpha\right] \\
    \leq \ & \left(V^{\lambda, \alpha}_{1}\left(s^\alpha, \pi^{\ast, \alpha}, \mu^{\mathcal{I}}_t \right) -V^{\lambda, \alpha}_{1}\left(s^\alpha, \pi^{\alpha}, \mu^{\mathcal{I}}_t \right) \right)\\
    & \ -\lambda  \EE_{\pi^{\ast, \alpha}, \mu^{\mathcal{I}}_t }\left[\sum^H_{h=1} \KL\left( \pi^{\ast, \alpha}(\cdot \mid s_h^\alpha) , \pi_{h}^{\alpha}(\cdot \mid s_h^\alpha)\right) \mid s_1^\alpha = s^\alpha\right] \,.
\end{align*}
Using the monotonicity condition, one can show that the first term is non-positive. Taking summation over $H$ and integrating over all agents, we have
\begin{align*}
& \int_\mathcal{I}\sum^H_{h=1} \EE_{\mu^{\ast,\alpha}_h}\left[\KL\left( \pi^{\ast, \alpha}(\cdot \mid s_h^\alpha) , \pi_{t+1,h}^{\alpha}(\cdot \mid s_h^\alpha)\right)   \right] d\nu(\alpha) \\
\leq \ & \left(1-\eta_t\lambda \right)\int_\mathcal{I}\sum^H_{h=1} \EE_{\mu^{\ast,\alpha}_h}\left[\KL\left( \pi^{\ast, \alpha}(\cdot \mid s_h^\alpha) , \pi_{t,h}^{\alpha}(\cdot \mid s_h^\alpha)\right)\right] d\nu(\alpha)  + \eta_t^2H^3 \,.
\end{align*}
The $\lambda$-regularization can also be interpreted as it regularize the game to be strongly convex with respect to KL divergence. As a result, one can anticipate the algorithm's convergence rate to be akin to established methods for strongly monotone games, often converging at a linear rate \citep{lin2020finite,cen2021fast,jordan2022adaptive}.

\paragraph{Tabular GMFG with bandit feedback}
We now consider the case where we only observe the cost $c_h(s,a,z)$ for the state, action, and state distribution that we have visited. In this case, we use importance sampling with implicit exploration (Eq.\ref{eq:tabular_approximation}) to estimate the gradient. We show that our algorithm then achieves the following last-iterate guarantee. 
\begin{thm}\label{thm:approx_Q}
Take $\eta_t = \frac{1}{t^{3/4}}$, $\gamma = \frac{1}{t^{1/4}}$, $\beta_t = \frac{H+1}{H+t}$. We have
    \begin{align*}
    D\left(\pi^\mathcal{I}_{t+1}\right) \leq  O \left(\frac{A\log(t)}{ t^{3/4}} + \frac{\sqrt{\log(A/\delta)}}{t^{3/4}} + \frac{H^3}{t^{1/4}} + \frac{\sqrt{H^3\iota}}{t^{1/4}} + \frac{\log(1/\delta)}{t^{1/2}}\right)\,.
\end{align*}
\end{thm}
We first fix a tuple $h, \alpha, s_h^\alpha$. Then, similar to Theorem \ref{thm:exact_Q}, using the mirror descent update, we can obtain the following relation on $\KL\left( \pi^{\ast, \alpha}(\cdot \mid s_h^\alpha) , \pi_{t+1,h}^{\alpha}(\cdot \mid s_h^\alpha)\right)$.
\begin{align*}
    & \KL\left( \pi^{\ast, \alpha}(\cdot \mid s_h^\alpha) , \pi_{t+1,h}^{\alpha}(\cdot \mid s_h^\alpha)\right) \\
     \leq \ & \KL\left( \pi^{\ast, \alpha}(\cdot \mid s_h^\alpha) , \pi_{t,h}^{\alpha}(\cdot \mid s_h^\alpha)\right) + \frac{\eta_t^2H^2}{2\gamma_t^2} \\
     & \ + \eta_t \left\langle \pi^{\ast, \alpha}(\cdot \mid s_h^\alpha) -\pi_{t,h}^{\alpha}(\cdot \mid s_h^\alpha) , Q^{\lambda, \alpha}_{h}\left(s_h^\alpha,\cdot, \pi^\beta_t, \mu^{\mathcal{I}}\right) +\lambda\log(\pi_{t,h}^{\alpha}(\cdot \mid s_h^\alpha)) \right\rangle  \\
     & \ + \eta_t \left\langle \pi^{\ast, \alpha}(\cdot \mid s_h^\alpha) -\pi_{t,h}^{\alpha}(\cdot \mid s_h^\alpha) , \hat{g}_{t,h}(s_h^\alpha, \cdot)- Q^{\lambda, \alpha}_{h}\left(s_h^\alpha,\cdot, \pi^\beta_t, \mu^{\mathcal{I}}\right) \right\rangle \,.
\end{align*}

As we use gradient estimation $\hat{g}_{t,h}$ instead of the exact gradient, we would need to characterize the estimation error $ \hat{g}_{t,h}(s_h^\alpha, \cdot)- Q^{\lambda, \alpha}_{h}\left(s_h^\alpha,\cdot, \pi^\beta_t, \mu^{\mathcal{I}}\right) $ to utilize our proof outline for Theorem \ref{thm:exact_Q}. This gradient estimation error can be further refined to the estimation error of the value function $\hat{V}^{\lambda,\alpha}_{h+1}\left(s_{h+1}^\alpha, \pi_t^\alpha, \mu^\mathcal{I}\right) - V^{\lambda,\alpha}_{h+1}\left(s_{h+1}^\alpha, \pi_t^\alpha, \mu^\mathcal{I}\right)$.

Using the update rule, we can upper bound the error as
\begin{align*}
        \hat{V}^{\lambda,\alpha}_{t,h}\left(s_h^\alpha, \pi_t^\alpha, \mu^\mathcal{I}\right) - V^{\lambda,\alpha}_{h}\left(s_h^\alpha, \pi_t^\alpha, \mu^\mathcal{I}\right) 
        \leq \ & \sum^k_{i=1}\EE_{\pi_{t^i}^\alpha}\left[ \beta_k^i \left(\left(P_h - \hat{P}_h^{t_i}\right)V^{\lambda,\alpha}_{h+1}\left( \pi_{t^i}^\alpha, \mu^\mathcal{I}\right)(s_h^\alpha,a^\alpha)\right)\right] \\
        & \ + \beta_k^i\left(\hat{V}^{\lambda,\alpha}_{t^i,h+1}\left(s^{t^i}_{h+1}, \pi_{t^i}^\alpha, \mu^\mathcal{I}\right) - V^{\lambda,\alpha}_{h+1}\left(s^{t^i}_{h+1}, \pi_{t^i}^\alpha, \mu^\mathcal{I}\right)\right)\,.
    \end{align*}
    Subsequently, using a martingale concentration inequality allows us to upper bound the first term at the order of $O \left(\sqrt{1/t}\right)$. Through an induction argument and leveraging the choice of learning rate $\beta_t$, we show that $\hat{V}^{\lambda,\alpha}_{t,h}\left(s_h^\alpha, \pi_t^\alpha, \mu^\mathcal{I}\right) - V^{\lambda,\alpha}_{h}\left(s_h^\alpha, \pi_t^\alpha, \mu^\mathcal{I}\right)$ is also upper bounded by $O \left(\sqrt{1/t}\right)$.

    Having characterized the estimation error, we then follow the same proof outline as outlined in Theorem \ref{thm:exact_Q} and obtain Theorem \ref{thm:approx_Q} by carefully choosing the parameters. 

\paragraph{Linear GMFG with bandit feedback}
Under a linearly parameterized GMFG, we use ridge regression to estimate the model parameter $\theta_h^\ast$ (Equation (\ref{eq:close_form})). Subsequently, we utilize value iteration alongside importance sampling with exploration to estimate the gradient (Equation (\ref{eq:linear_approx_Q})). Our algorithm provides the following convergence guarantee for the last iterate.
\begin{thm}\label{thm:linear_Q}
    Take $\eta_t = \frac{1}{t^{4/5}}, \gamma_t = \frac{1}{t^{1/5}}$. We have
    \begin{align*}
         D\left(\pi^\mathcal{I}_{t+1}\right)\leq O \left(\frac{H^3}{t^{3/5}} + \frac{\sqrt{\lambda}d^2H^2\log(t+1/d\tau)}{t^{1/5}} + \frac{A\log(t)}{ t^{1/5}} + \sqrt{\frac{\log(A/\delta)}{t^{2/5}}}\right) \,.
    \end{align*}
\end{thm}

Similar to the analysis of Theorem \ref{thm:approx_Q}, the key to deriving Theorem \ref{thm:linear_Q} lies in characterizing the estimation error $ \hat{g}_{t,h}(s_h^\alpha, a^\alpha)$. This is then affected by the estimation error of our value function, and the estimation of $\theta_h$. Leveraging a uniform convergence lemma (Lemma \ref{lem:8_7_agarwal}), we can effectively upper bound the error as:
\begin{align*}
     \hat{Q}^{\lambda, \alpha}_{h}\left(s_h^\alpha, a_h^\alpha, \pi^\beta_t, \mu^{\mathcal{I}}\right) - Q^{\lambda, \alpha}_{h}\left(s_h^\alpha,a_h^\alpha, \pi^\beta_t, \mu^{\mathcal{I}}\right)
     = \ & \phi(s_h^\alpha, a_h^\alpha, z_h^\alpha)^\top \left(\hat{\theta}_{t,h} - \theta_h^\ast\right)^\top V^{\lambda, \alpha}_{h + 1} \left(s_{h + 1}^\alpha, \pi^\alpha, \mu^{\mathcal{I}}\right)\\
     \leq \ & \Tilde{O}\left(dH \|\phi(s^\alpha_h,a^\alpha_h,z_h^\alpha)\|_{\Lambda_{t,h}^{-1}}\right)\,.
\end{align*}
Employing the same argument as Theorem \ref{thm:exact_Q}, we obtain a recursion on $D(\pi_t)$, which involves a summation of this estimation error $\int_\mathcal{I}\sum^H_{h=1} \sum^t_{k=1}\eta_k \omega_t^k\EE_{\mu^{\ast,\alpha}_h}\left[dH \|\phi(s^\alpha_{k,h},a^\alpha_{k,h}, z_{k,h}^\alpha)\|_{\Lambda_{t,h}^{-1}}\right] d\nu(\alpha) $.   
Lastly, we bound the summation of this term using an elliptical potential lemma and recursively apply the relationship of $D(\pi_t)$ to obtain the final bound.

\section{Experiments}
To verify the effectiveness of our algorithm, we analyze it empirically on four environments, Predator Prey, Crowd Avoidance, Crowd modeling, and Periodic Aversion. To assess the performance of our algorithm, we use exploitability, which is defined as $\text{exploitability}(\pi) = \int_\mathcal{I} J(\pi, \mu^\mathcal{I}) - \min_{\pi^\prime}J(\pi^\prime, \mu^\mathcal{I})$. To ensure reproducibility, we repeat each set of environments with 5 random seeds and present the result with one standard deviation.

\subsection{Basline algorithm}
For the baseline algorithm, we use fictitious play, a method known for providing a robust approximation of Nash equilibrium. Fictitious play iteratively computes the best response against the distribution induced by averaging past best responses. This method is known to perform well in various games, including mean-field games. However, since it needs to compute the best response every iteration, it is much more computationally heavier than our algorithm. We use the setup of fictitious play from \cite{perrin2020fictitious} and follow the implementation from OpenSpiel.
 
\subsection{Enviornment set ups}
For all the environments described below, we use the OpenSpiel implementation of the games. 

\paragraph{Crowd Modeling} The Crowd Modeling game also referred to as the beach bar process, presents a simplified rendition of the renowned Santa Fe bar problem \citep{greenwald1997learning}. Following the dynamic modeling and cost functions outlined in \cite{perrin2020fictitious} (section 4.2), we conduct our experiment with 10 states and 3 actions. A beach bar is located in one of the states. In a scorching weather condition, agents aim to position themselves in proximity to the bar while avoiding excessively crowded areas.

\paragraph{Crowd Avoidance}
The crowd avoidance problem is a simple two-population game. The game has $7$ states and $5$ actions (stay, up, down, left, and right). The agents will receive a cost of $1$ if they collide with each other, and a cost of $0$ if otherwise. We follow the implementation from OpenSpiel. While trying to avoid congestion, the agent must also avoid the forbidden states. 

\paragraph{Predator Prey}
We follow the setup described in section 5.4 of \cite{perolat2022scaling}. This game features three populations and bears a close resemblance to the popular outdoor game for children, Hens-Foxes-Snakes. In this context, hens endeavor to capture snakes, snakes pursue foxes, and foxes are inclined to prey upon hens. Although the population sizes are predetermined, the cost structure incentivizes agents to chase the population they dominate.

\paragraph{Periodic Aversion}
The periodic aversion game was first introduced in \cite{almulla2017two} and served as is an approximation of a continuous space, continuous time model introduced to study ergodic MFG with an explicit solution. We follow the implementation from OpenSpiel and the description from \cite{elie2020convergence}. Each agent has a position on a torus $T=[0,1]$ with periodic boundary conditions. The cost is then determined by a combination of the current position on the torus, the action, and the congestion of the agents.  

\subsection{Parameters and experiment configuration}
For all of our experiments, we choose the learning rate to be $\eta_t = 0.1$ and the exploration rate $\gamma_t = 0.1$. We repeat the experiments with $5$ different random seeds. We ran all experiments with a 10-core CPU, with 32 GB memory. 

\subsection{Experimental results}
We show the results of the four environments described above. As evident in Figure  \ref{fig}, the mirror descent algorithm attains comparable performance as the fictitious play in two out of four environments, while enjoying much better computational complexity as it does not require the computation of the best response at each iteration. By approximating the value function with a linear structure, the mirror descent algorithm gains improvement in the results in two out of four environments. 
%{\color{blue}todo: add parameter description here.}
\begin{figure}[h]
    \centering
    \includegraphics[width=0.45\textwidth]{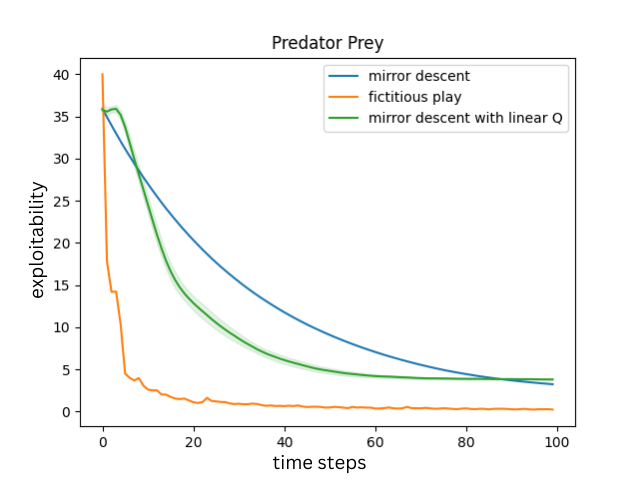}
    \includegraphics[width=0.45\textwidth]{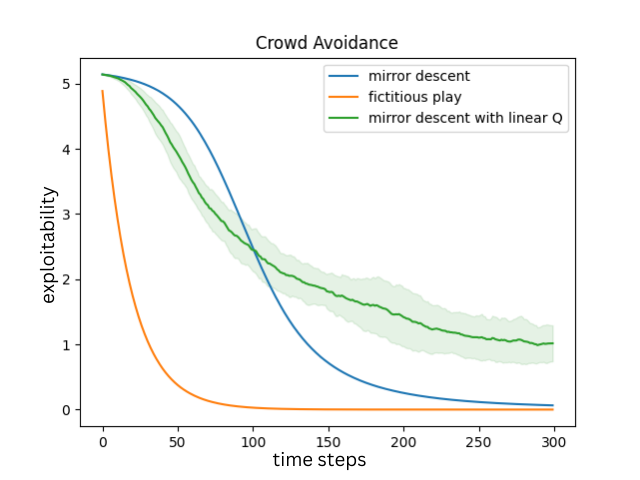}
     \includegraphics[width=0.45\textwidth]{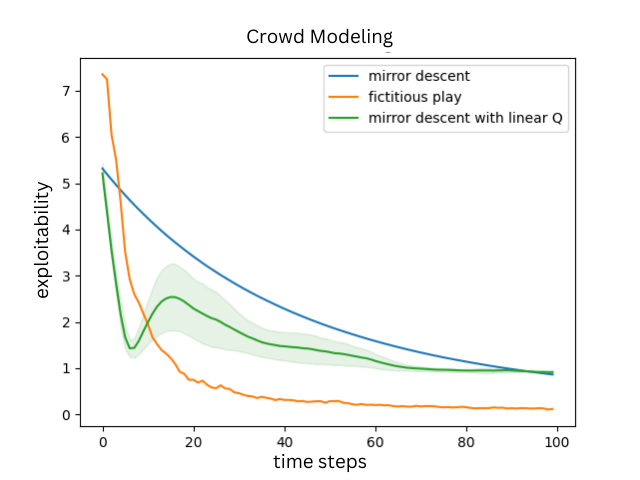}
     \includegraphics[width=0.45\textwidth]{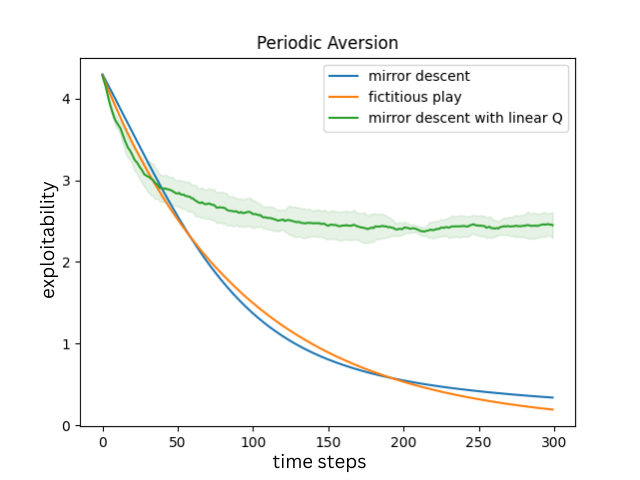}
    \caption{Experimental results for the mean field games described. }\label{fig}

\end{figure}

\section{Conclusion}

In this work, we present the first last-iterate convergence rate for monotone GMFGs with a mirror-descent algorithm. In tabular monotone GMFGs and under bandit feedback, we obtain a $O(T^{-1/4})$ last-iterate convergence rate. Under access to the exact cost and transition functions, we improved the rate to $O(T^{-1})$. In linear GMFGs, we achieve a last-iterate convergence rate of $O(T^{-1/5})$ under bandit feedback. Our study improves the understanding of mean-field games and the commonly used algorithms by providing insights both theoretically and numerically.

\newpage
\bibliography{ref}

\begin{thebibliography}{}

\bibitem[Achdou et~al., 2012]{achdou2012mean}
Achdou, Y., Camilli, F., and Capuzzo-Dolcetta, I. (2012).
\newblock Mean field games: numerical methods for the planning problem.
\newblock {\em SIAM Journal on Control and Optimization}, 50(1):77--109.

\bibitem[Achdou and Capuzzo-Dolcetta, 2010]{achdou2010mean}
Achdou, Y. and Capuzzo-Dolcetta, I. (2010).
\newblock Mean field games: numerical methods.
\newblock {\em SIAM Journal on Numerical Analysis}, 48(3):1136--1162.

\bibitem[Achdou et~al., 2020]{achdou2020mean}
Achdou, Y., Cardaliaguet, P., Delarue, F., Porretta, A., Santambrogio, F., Achdou, Y., and Lauri{\`e}re, M. (2020).
\newblock Mean field games and applications: Numerical aspects.
\newblock {\em Mean Field Games: Cetraro, Italy 2019}, pages 249--307.

\bibitem[Agarwal et~al., 2019]{agarwal2019reinforcement}
Agarwal, A., Jiang, N., Kakade, S.~M., and Sun, W. (2019).
\newblock Reinforcement learning: Theory and algorithms.
\newblock {\em CS Dept., UW Seattle, Seattle, WA, USA, Tech. Rep}, 32:96.

\bibitem[Almulla et~al., 2017]{almulla2017two}
Almulla, N., Ferreira, R., and Gomes, D. (2017).
\newblock Two numerical approaches to stationary mean-field games.
\newblock {\em Dynamic Games and Applications}, 7:657--682.

\bibitem[Aurell et~al., 2022]{aurell2022finite}
Aurell, A., Carmona, R., Dayan{\i}kl{\i}, G., and Lauri{\`e}re, M. (2022).
\newblock Finite state graphon games with applications to epidemics.
\newblock {\em Dynamic Games and Applications}, 12(1):49--81.

\bibitem[Avena-Koenigsberger et~al., 2018]{avena2018communication}
Avena-Koenigsberger, A., Misic, B., and Sporns, O. (2018).
\newblock Communication dynamics in complex brain networks.
\newblock {\em Nature reviews neuroscience}, 19(1):17--33.

\bibitem[Bai et~al., 2020]{bai2020near}
Bai, Y., Jin, C., and Yu, T. (2020).
\newblock Near-optimal reinforcement learning with self-play.
\newblock {\em Advances in neural information processing systems}.

\bibitem[Bakker et~al., 2010]{bakker2010social}
Bakker, L., Hare, W., Khosravi, H., and Ramadanovic, B. (2010).
\newblock A social network model of investment behaviour in the stock market.
\newblock {\em Physica A: Statistical Mechanics and its Applications}, 389(6):1223--1229.

\bibitem[Bervoets et~al., 2020]{bervoets2020learning}
Bervoets, S., Bravo, M., and Faure, M. (2020).
\newblock Learning with minimal information in continuous games.
\newblock {\em Theoretical Economics}, 15(4):1471--1508.

\bibitem[Bian et~al., 2016]{bian2016evolving}
Bian, Y.-t., Xu, L., and Li, J.-s. (2016).
\newblock Evolving dynamics of trading behavior based on coordination game in complex networks.
\newblock {\em Physica A: Statistical Mechanics and its Applications}, 449:281--290.

\bibitem[Bravo et~al., 2018]{bravo2018bandit}
Bravo, M., Leslie, D., and Mertikopoulos, P. (2018).
\newblock Bandit learning in concave n-person games.
\newblock {\em Advances in Neural Information Processing Systems}.

\bibitem[Bullmore and Sporns, 2009]{bullmore2009complex}
Bullmore, E. and Sporns, O. (2009).
\newblock Complex brain networks: graph theoretical analysis of structural and functional systems.
\newblock {\em Nature reviews neuroscience}, 10(3):186--198.

\bibitem[Bullmore and Sporns, 2012]{bullmore2012economy}
Bullmore, E. and Sporns, O. (2012).
\newblock The economy of brain network organization.
\newblock {\em Nature reviews neuroscience}, 13(5):336--349.

\bibitem[Cai et~al., 2023]{cai2023uncoupled}
Cai, Y., Luo, H., Wei, C.-Y., and Zheng, W. (2023).
\newblock Uncoupled and convergent learning in two-player zero-sum markov games.
\newblock In {\em Advances in Neural Information Processing Systems}.

\bibitem[Caines et~al., 2006]{caines2006large}
Caines, P.~E., Huang, M., and Malham{\'e}, R.~P. (2006).
\newblock Large population stochastic dynamic games: closed-loop mckean-vlasov systems and the nash certainty equivalence principle.
\newblock {\em Communications in Information and Systems}, 6(3):221--252.

\bibitem[Cen et~al., 2021]{cen2021fast}
Cen, S., Wei, Y., and Chi, Y. (2021).
\newblock Fast policy extragradient methods for competitive games with entropy regularization.
\newblock {\em Advances in Neural Information Processing Systems}.

\bibitem[Cui and Koeppl, 2021a]{cui2021approximately}
Cui, K. and Koeppl, H. (2021a).
\newblock Approximately solving mean field games via entropy-regularized deep reinforcement learning.
\newblock In {\em International Conference on Artificial Intelligence and Statistics}.

\bibitem[Cui and Koeppl, 2021b]{cui2021learning}
Cui, K. and Koeppl, H. (2021b).
\newblock Learning graphon mean field games and approximate nash equilibria.
\newblock In {\em International Conference on Learning Representations}.

\bibitem[Cui et~al., 2023]{cui2023scalable}
Cui, K., Li, M., Fabian, C., and Koeppl, H. (2023).
\newblock Scalable task-driven robotic swarm control via collision avoidance and learning mean-field control.
\newblock In {\em International Conference on Robotics and Automation (ICRA)}.

\bibitem[Duvocelle et~al., 2023]{duvocelle2023multiagent}
Duvocelle, B., Mertikopoulos, P., Staudigl, M., and Vermeulen, D. (2023).
\newblock Multiagent online learning in time-varying games.
\newblock {\em Mathematics of Operations Research}, 48(2):914--941.

\bibitem[Elie et~al., 2020]{elie2020convergence}
Elie, R., Perolat, J., Lauri{\`e}re, M., Geist, M., and Pietquin, O. (2020).
\newblock On the convergence of model free learning in mean field games.
\newblock In {\em The AAAI Conference on Artificial Intelligence}.

\bibitem[Fabian et~al., 2023]{fabian2023learning}
Fabian, C., Cui, K., and Koeppl, H. (2023).
\newblock Learning sparse graphon mean field games.
\newblock In {\em International Conference on Artificial Intelligence and Statistics}.

\bibitem[Geist et~al., 2022]{geist2022concave}
Geist, M., P{\'e}rolat, J., Lauri{\`e}re, M., Elie, R., Perrin, S., Bachem, O., Munos, R., and Pietquin, O. (2022).
\newblock Concave utility reinforcement learning: The mean-field game viewpoint.
\newblock In {\em International Conference on Autonomous Agents and Multiagent Systems}.

\bibitem[Greenwald et~al., 1997]{greenwald1997learning}
Greenwald, A., Mishra, B., and Parikh, R. (1997).
\newblock Learning in the santa fe bar problem.

\bibitem[Jin et~al., 2018]{jin2018q}
Jin, C., Allen-Zhu, Z., Bubeck, S., and Jordan, M.~I. (2018).
\newblock Is q-learning provably efficient?
\newblock {\em Advances in neural information processing systems}.

\bibitem[Jin et~al., 2022]{jin2022v}
Jin, C., Liu, Q., Wang, Y., and Yu, T. (2022).
\newblock V-learning--a simple, efficient, decentralized algorithm for multiagent rl.
\newblock In {\em ICLR 2022 Workshop on Gamification and Multiagent Solutions}.

\bibitem[Jordan et~al., 2022]{jordan2022adaptive}
Jordan, M.~I., Lin, T., and Zhou, Z. (2022).
\newblock Adaptive, doubly optimal no-regret learning in games with gradient feedback.
\newblock {\em Games with Gradient Feedback (September 8, 2022)}.

\bibitem[Lasry and Lions, 2007]{lasry2007mean}
Lasry, J.-M. and Lions, P.-L. (2007).
\newblock Mean field games.
\newblock {\em Japanese journal of mathematics}, 2(1):229--260.

\bibitem[Lauriere et~al., 2022]{lauriere2022scalable}
Lauriere, M., Perrin, S., Girgin, S., Muller, P., Jain, A., Cabannes, T., Piliouras, G., P{\'e}rolat, J., Elie, R., Pietquin, O., et~al. (2022).
\newblock Scalable deep reinforcement learning algorithms for mean field games.
\newblock In {\em International Conference on Machine Learning}.

\bibitem[Liang and Stokes, 2019]{liang2019interaction}
Liang, T. and Stokes, J. (2019).
\newblock Interaction matters: A note on non-asymptotic local convergence of generative adversarial networks.
\newblock In {\em International Conference on Artificial Intelligence and Statistics}.

\bibitem[Lin et~al., 2021]{lin2021doubly}
Lin, T., Zhou, Z., Ba, W., and Zhang, J. (2021).
\newblock Doubly optimal no-regret online learning in strongly monotone games with bandit feedback.
\newblock {\em arXiv preprint arXiv:2112.02856}.

\bibitem[Lin et~al., 2020]{lin2020finite}
Lin, T., Zhou, Z., Mertikopoulos, P., and Jordan, M. (2020).
\newblock Finite-time last-iterate convergence for multi-agent learning in games.
\newblock In {\em International Conference on Machine Learning}.

\bibitem[Mertikopoulos et~al., 2018]{mertikopoulos2018cycles}
Mertikopoulos, P., Papadimitriou, C., and Piliouras, G. (2018).
\newblock Cycles in adversarial regularized learning.
\newblock In {\em Symposium on discrete algorithms}.

\bibitem[Newman, 2002]{newman2002spread}
Newman, M.~E. (2002).
\newblock Spread of epidemic disease on networks.
\newblock {\em Physical review E}, 66(1):016128.

\bibitem[Parise and Ozdaglar, 2019]{parise2019graphon}
Parise, F. and Ozdaglar, A. (2019).
\newblock Graphon games.
\newblock In {\em Conference on Economics and Computation}.

\bibitem[Pastor-Satorras et~al., 2015]{pastor2015epidemic}
Pastor-Satorras, R., Castellano, C., Van~Mieghem, P., and Vespignani, A. (2015).
\newblock Epidemic processes in complex networks.
\newblock {\em Reviews of modern physics}, 87(3):925.

\bibitem[P{\'e}rolat et~al., 2022]{perolat2022scaling}
P{\'e}rolat, J., Perrin, S., Elie, R., Lauri{\`e}re, M., Piliouras, G., Geist, M., Tuyls, K., and Pietquin, O. (2022).
\newblock Scaling mean field games by online mirror descent.
\newblock In {\em International Conference on Autonomous Agents and Multiagent Systems}.

\bibitem[Perrin et~al., 2020]{perrin2020fictitious}
Perrin, S., P{\'e}rolat, J., Lauri{\`e}re, M., Geist, M., Elie, R., and Pietquin, O. (2020).
\newblock Fictitious play for mean field games: Continuous time analysis and applications.
\newblock {\em Advances in neural information processing systems}.

\bibitem[Shani et~al., 2020]{shani2020adaptive}
Shani, L., Efroni, Y., and Mannor, S. (2020).
\newblock Adaptive trust region policy optimization: Global convergence and faster rates for regularized mdps.
\newblock In {\em The AAAI Conference on Artificial Intelligence}.

\bibitem[Tangpi and Zhou, 2024]{tangpi2024optimal}
Tangpi, L. and Zhou, X. (2024).
\newblock Optimal investment in a large population of competitive and heterogeneous agents.
\newblock {\em Finance and Stochastics}, pages 1--55.

\bibitem[Tseng, 1995]{tseng1995linear}
Tseng, P. (1995).
\newblock On linear convergence of iterative methods for the variational inequality problem.
\newblock {\em Journal of Computational and Applied Mathematics}, 60(1-2):237--252.

\bibitem[Yang et~al., 2018]{yang2018learning}
Yang, J., Ye, X., Trivedi, R., Xu, H., and Zha, H. (2018).
\newblock Learning deep mean field games for modeling large population behavior.
\newblock In {\em International Conference on Learning Representations}.

\bibitem[Zhang et~al., 2023]{zhang2023learning}
Zhang, F., Tan, V., Wang, Z., and Yang, Z. (2023).
\newblock Learning regularized monotone graphon mean-field games.
\newblock {\em Advances in Neural Information Processing Systems}.

\bibitem[Zhou et~al., 2020]{zhou2020convergence}
Zhou, Z., Mertikopoulos, P., Bambos, N., Boyd, S.~P., and Glynn, P.~W. (2020).
\newblock On the convergence of mirror descent beyond stochastic convex programming.
\newblock {\em SIAM Journal on Optimization}, 30(1):687--716.

\end{thebibliography}

%%%%%%%%%%%%%%%%%%%%%%%%%%%%%%%%%%%%%%%%%%%%%%%%%%%%%%%%%%%%

\appendix

%%%%%%%%%%%%%%%%%%%%%%%%%%%%%%%%%%%%%%%%%%%%%%%%%%%%%%%%%%%%

\newpage

\section{Proof of Theorem \ref{thm:exact_Q}}

\begin{proof}
By the update rule, we have
\begin{align*}
    & \KL\left( \pi^{\ast, \alpha}(\cdot \mid s_h^\alpha) , \pi_{t+1,h}^{\alpha}(\cdot \mid s_h^\alpha)\right) \\
    = \ & \KL\left( \pi^{\ast, \alpha}(\cdot \mid s_h^\alpha) , \pi_{t,h}^{\alpha}(\cdot \mid s_h^\alpha)\right) - \KL\left(\pi_{t+1,h}^{\alpha}(\cdot \mid s_h^\alpha) , \pi_{t,h}^{\alpha}(\cdot \mid s_h^\alpha)\right) \\
    & \ +\eta_t \left\langle \pi^{\ast, \alpha}(\cdot \mid s_h^\alpha) - \pi_{t+1,h}^{\alpha}(\cdot \mid s_h^\alpha) , Q^{\lambda, \alpha}_{h}\left(s_h^\alpha,\cdot, \pi^\beta_t, \mu^{\mathcal{I}}\right) +\lambda\log(\pi_{t,h}^{\alpha}(\cdot \mid s_h^\alpha)) \right\rangle \\
    = \ & \KL\left( \pi^{\ast, \alpha}(\cdot \mid s_h^\alpha) , \pi_{t,h}^{\alpha}(\cdot \mid s_h^\alpha)\right) - \KL\left(\pi_{t+1,h}^{\alpha}(\cdot \mid s_h^\alpha) , \pi_{t,h}^{\alpha}(\cdot \mid s_h^\alpha)\right) \\
    & \ +\eta_t \left\langle \pi^{\ast, \alpha}(\cdot \mid s_h^\alpha) - \pi_{t,h}^{\alpha}(\cdot \mid s_h^\alpha) , Q^{\lambda, \alpha}_{h}\left(s_h^\alpha,\cdot, \pi^\beta_t, \mu^{\mathcal{I}}\right) +\lambda\log(\pi_{t,h}^{\alpha}(\cdot \mid s_h^\alpha)) \right\rangle \\
    & \ +\eta_t \left\langle \pi_{t,h}^{\alpha}(\cdot \mid s_h^\alpha) - \pi_{t+1,h}^{\alpha}(\cdot \mid s_h^\alpha), Q^{\lambda, \alpha}_{h}\left(s_h^\alpha,\cdot, \pi^\beta_t, \mu^{\mathcal{I}}\right) +\lambda\log(\pi_{t,h}^{\alpha}(\cdot \mid s_h^\alpha)) \right\rangle  \\
    \leq \ & \KL\left( \pi^{\ast, \alpha}(\cdot \mid s_h^\alpha) , \pi_{t,h}^{\alpha}(\cdot \mid s_h^\alpha)\right) - \KL\left(\pi_{t+1,h}^{\alpha}(\cdot \mid s_h^\alpha) , \pi_{t,h}^{\alpha}(\cdot \mid s_h^\alpha)\right)\\
    & \ +\eta_t \left\| \pi_{t+1,h}^{\alpha}(\cdot \mid s_h^\alpha) - \pi_{t,h}^{\alpha}(\cdot \mid s_h^\alpha)\right\|_1\left\| Q^{\lambda, \alpha}_{h}\left(s_h^\alpha,\cdot, \pi^\beta_t, \mu^{\mathcal{I}}\right) \right\|_\infty \\
    & \ + \eta_t \left\langle \pi^{\ast, \alpha}(\cdot \mid s_h^\alpha) - \pi_{t,h}^{\alpha}(\cdot \mid s_h^\alpha)  , Q^{\lambda, \alpha}_{h}\left(s_h^\alpha,\cdot, \pi^\beta_t, \mu^{\mathcal{I}}\right) +\lambda\log(\pi_{t,h}^{\alpha}(\cdot \mid s_h^\alpha)) \right\rangle \\
    \leq \ & \KL\left( \pi^{\ast, \alpha}(\cdot \mid s_h^\alpha) , \pi_{t,h}^{\alpha}(\cdot \mid s_h^\alpha)\right) - \KL\left(\pi_{t+1,h}^{\alpha}(\cdot \mid s_h^\alpha) , \pi_{t,h}^{\alpha}(\cdot \mid s_h^\alpha)\right) +\eta_t H\sqrt{2\KL\left( \pi_{t+1,h}^{\alpha}(\cdot \mid s_h^\alpha),  \pi_{t,h}^{\alpha}(\cdot \mid s_h^\alpha)\right)}\\
     & \ + \eta_t \left\langle \pi^{\ast, \alpha}(\cdot \mid s_h^\alpha) - \pi_{t,h}^{\alpha}(\cdot \mid s_h^\alpha)  , Q^{\lambda, \alpha}_{h}\left(s_h^\alpha,\cdot, \pi^\beta_t, \mu^{\mathcal{I}}\right) +\lambda\log(\pi_{t,h}^{\alpha}(\cdot \mid s_h^\alpha)) \right\rangle  \\
     \leq \ & \KL\left( \pi^{\ast, \alpha}(\cdot \mid s_h^\alpha) , \pi_{t,h}^{\alpha}(\cdot \mid s_h^\alpha)\right) + \frac{\eta_t^2H^2}{2} + \eta_t \left\langle \pi^{\ast, \alpha}(\cdot \mid s_h^\alpha) - \pi_{t,h}^{\alpha}(\cdot \mid s_h^\alpha)  , Q^{\lambda, \alpha}_{h}\left(s_h^\alpha,\cdot, \pi^\beta_t, \mu^{\mathcal{I}}\right) +\lambda\log(\pi_{t,h}^{\alpha}(\cdot \mid s_h^\alpha)) \right\rangle  \,.
\end{align*}

Notice that $V^{\lambda, \alpha}_{h}\left(s_h^\alpha, \pi^{\alpha}, \mu^{\mathcal{I}}_t \right) = \left\langle Q^{\lambda, \alpha}_{h}\left(s_h^\alpha, \cdot, \pi^{\alpha}, \mu^{\mathcal{I}}_t \right) + \lambda \log(\pi^{\alpha}(\cdot \mid s_h)), \pi^{\alpha}(\cdot \mid s_h)\right\rangle$. 
By the definition of the value function and the performance difference lemma, we have
\begin{align*}
    & V^{\lambda, \alpha}_{1}\left(s^\alpha, \pi^{\ast, \alpha}, \mu^{\mathcal{I}}_t \right) -V^{\lambda, \alpha}_{1}\left(s^\alpha, \pi^{\alpha}, \mu^{\mathcal{I}}_t \right) \\
    = \ & \EE_{\pi^{\ast, \alpha}, \mu^{\mathcal{I}}_t }\left[\sum^H_{h=1} c_h(s_h^\alpha, a_h^\alpha, z_h^\alpha) +\lambda\log(\pi_{h}^{\ast,\alpha}(\cdot \mid s_h^\alpha))  \mid s_1 = s\right]- V^{\lambda, \alpha}_{1}\left(s^\alpha, \pi^{\alpha}, \mu^{\mathcal{I}}_t \right)\\
    = \ & \EE_{\pi^{\ast, \alpha}, \mu^{\mathcal{I}}_t }\left[\sum^H_{h=1} c_h(s_h^\alpha, a_h^\alpha, z_h^\alpha) +\lambda\log(\pi_{h}^{\ast,\alpha}(\cdot \mid s_h^\alpha)) + V^{\lambda, \alpha}_{h}\left(s_h^\alpha, \pi^{ \alpha}, \mu^{\mathcal{I}}_t \right)- V^{\lambda, \alpha}_{h}\left(s_h^\alpha, \pi^{ \alpha}, \mu^{\mathcal{I}}_t \right) \mid s_1^\alpha = s^\alpha\right] \\
    & \ - V^{\lambda, \alpha}_{1}\left(s^\alpha, \pi^{\alpha}, \mu^{\mathcal{I}}_t \right)\\
    = \ & \EE_{\pi^{\ast, \alpha}, \mu^{\mathcal{I}}_t }\left[\sum^H_{h=1} c_h(s_h^\alpha, a_h^\alpha, z_h^\alpha) + V^{\lambda, \alpha}_{h+1}\left(s_{h+1}^\alpha, \pi^{\alpha}, \mu^{\mathcal{I}}_t \right) +\lambda\log(\pi_{h}^{\ast,\alpha}(\cdot \mid s_h^\alpha)) - V^{\lambda, \alpha}_{h}\left(s_h^\alpha, \pi^{ \alpha}, \mu^{\mathcal{I}}_t \right) \mid s_1^\alpha = s^\alpha\right] \\
     = \ &  \EE_{\pi^{\ast, \alpha}, \mu^{\mathcal{I}}_t }\left[\sum^H_{h=1}  \left\langle Q^{\lambda, \alpha}_{h}\left(s_{h}^\alpha,\cdot, \pi^{\alpha} , \mu^{\mathcal{I}}_t \right) +\lambda\log(\pi_{h}^{\ast,\alpha}(\cdot \mid s_h^\alpha)) , \pi_{h}^{\ast,\alpha}(\cdot \mid s_h^\alpha)\right\rangle- V^{\lambda, \alpha}_{h}\left(s_h^\alpha, \pi^{\alpha}, \mu^{\mathcal{I}}_t \right) \mid s_1^\alpha = s^\alpha\right] \\
     = \ & \EE_{\pi^{\ast, \alpha}, \mu^{\mathcal{I}}_t }\left[\sum^H_{h=1}  \left\langle Q^{\lambda, \alpha}_{h}\left(s_{h}^\alpha,\cdot, \pi^{\alpha} , \mu^{\mathcal{I}}_t \right) , \pi_{h}^{\ast,\alpha}(\cdot \mid s_h^\alpha) - \pi_{h}^{\alpha}(\cdot \mid s_h^\alpha)\right\rangle \mid s_1^\alpha = s^\alpha\right]\\
     & \ +  \lambda \EE_{\pi^{\ast, \alpha}, \mu^{\mathcal{I}}_t}\left[\sum^H_{h=1}  \left\langle \log\left(\pi_{h}^{\ast,\alpha}(\cdot \mid s_h^\alpha)\right) , \pi_{h}^{\ast,\alpha}(\cdot \mid s_h^\alpha) \right\rangle- \left\langle \log(\pi^{\alpha}(\cdot \mid s_h^\alpha)), \pi^{\alpha}(\cdot \mid s_h^\alpha)\right\rangle\mid s_1^\alpha = s^\alpha\right] \\
     = \ & \EE_{\pi^{\ast, \alpha}, \mu^{\mathcal{I}}_t }\left[\sum^H_{h=1}  \left\langle Q^{\lambda, \alpha}_{h}\left(s_{h}^\alpha,\cdot, \pi^{\alpha} , \mu^{\mathcal{I}}_t \right) + \lambda\log(\pi^{\alpha}(\cdot \mid s_h^\alpha)), \pi_{h}^{\ast,\alpha}(\cdot \mid s_h^\alpha) - \pi_{h}^{\alpha}(\cdot \mid s_h^\alpha)\right\rangle  \mid s_1^\alpha = s^\alpha\right]\\
    & \ + \lambda  \EE_{\pi^{\ast, \alpha}, \mu^{\mathcal{I}}_t }\left[\sum^H_{h=1} \KL\left( \pi^{\ast, \alpha}(\cdot \mid s_h^\alpha) , \pi_{h}^{\alpha}(\cdot \mid s_h^\alpha)\right) \mid s_1^\alpha = s^\alpha\right]\,.
\end{align*}

Summing over $H$ and taking integral over $\alpha$, we therefore have
\begin{align*}
& \int_\mathcal{I}\sum^H_{h=1} \EE_{\mu^{\ast,\alpha}_h}\left[\KL\left( \pi^{\ast, \alpha}(\cdot \mid s_h^\alpha) , \pi_{t+1,h}^{\alpha}(\cdot \mid s_h^\alpha)\right)   \right] d\nu(\alpha)\\
\leq \ & \left(1-\eta_t\lambda \right)\int_\mathcal{I}\sum^H_{h=1} \EE_{\mu^{\ast,\alpha}_h}\left[\KL\left( \pi^{\ast, \alpha}(\cdot \mid s_h^\alpha) , \pi_{t,h}^{\alpha}(\cdot \mid s_h^\alpha)\right)\right]d\nu(\alpha)  + \eta_t^2H^3\\
& \ + \eta_t\int_\mathcal{I}\left(J^{\lambda, \alpha} (\pi^{\ast, \alpha}, \mu_t^\mathcal{I}) - J^{\lambda, \alpha} (\pi^{\alpha}, \mu_t^\mathcal{I})\right)d\nu(\alpha) \\
\leq \ & \left(1-\eta_t\lambda \right)\int_\mathcal{I}\sum^H_{h=1} \EE_{\mu^{\ast,\alpha}_h}\left[\KL\left( \pi^{\ast, \alpha}(\cdot \mid s_h^\alpha) , \pi_{t,h}^{\alpha}(\cdot \mid s_h^\alpha)\right)\right] d\nu(\alpha)  + \eta_t^2H^3 \,,
\end{align*}
as by proposition \ref{prop:monotone}, $\int_\mathcal{I}J^{\lambda, \alpha} (\pi^{\ast, \alpha}, \mu_t^\mathcal{I}) - J^{\lambda, \alpha} (\pi^{\alpha}, \mu_t^\mathcal{I}) d\nu(\alpha)\leq \int_\mathcal{I} J^{\lambda, \alpha} (\pi^{\ast, \alpha}, \mu_t^{\ast,\mathcal{I}}) - J^{\lambda, \alpha} (\pi^{\alpha}, \mu_t^{\ast,\mathcal{I}}) d\nu(\alpha)\leq 0$, by optimality.

Thus, let $\eta_t = t^{-1}$, we have
\begin{align*}
    D\left(\pi^\mathcal{I}_{t+1}\right)
    \leq  \ & (1-\eta_t\lambda) D\left(\pi^\mathcal{I}_{t}\right) + \eta_t^2 H^3
    \leq  \frac{H^3}{\lambda t}\,.
\end{align*}

\end{proof}

\section{Proof of Theorem \ref{thm:approx_Q}}

\begin{proof}
By the update rule, we have
\begin{align*}
    & \KL\left( \pi^{\ast, \alpha}(\cdot \mid s_h^\alpha) , \pi_{t+1,h}^{\alpha}(\cdot \mid s_h^\alpha)\right) \\
    = \ & \KL\left( \pi^{\ast, \alpha}(\cdot \mid s_h^\alpha) , \pi_{t,h}^{\alpha}(\cdot \mid s_h^\alpha)\right) - \KL\left(\pi_{t+1,h}^{\alpha}(\cdot \mid s_h^\alpha) , \pi_{t,h}^{\alpha}(\cdot \mid s_h^\alpha)\right) \\
    & \ +\eta_t \left\langle \pi^{\ast, \alpha}(\cdot \mid s_h^\alpha) -\pi_{t+1,h}^{\alpha}(\cdot \mid s_h^\alpha) , \hat{g}_{t,h}(s_h^\alpha, \cdot) +\lambda\log(\pi_{t,h}^{\alpha}(\cdot \mid s_h^\alpha)) \right\rangle \\
    = \ & \KL\left( \pi^{\ast, \alpha}(\cdot \mid s_h^\alpha) , \pi_{t,h}^{\alpha}(\cdot \mid s_h^\alpha)\right) - \KL\left(\pi_{t+1,h}^{\alpha}(\cdot \mid s_h^\alpha) , \pi_{t,h}^{\alpha}(\cdot \mid s_h^\alpha)\right) \\
    & \ +\eta_t \left\langle\pi^{\ast, \alpha}(\cdot \mid s_h^\alpha) -\pi_{t,h}^{\alpha}(\cdot \mid s_h^\alpha) , \hat{g}_{t,h}(s_h^\alpha, \cdot)+\lambda\log(\pi_{t,h}^{\alpha}(\cdot \mid s_h^\alpha)) \right\rangle \\
    & \ +\eta_t \left\langle \pi_{t,h}^{\alpha}(\cdot \mid s_h^\alpha) - \pi_{t+1,h}^{\alpha}(\cdot \mid s_h^\alpha), \hat{g}_{t,h}(s_h^\alpha, \cdot)+\lambda\log(\pi_{t,h}^{\alpha}(\cdot \mid s_h^\alpha)) \right\rangle  \\
    \leq \ & \KL\left( \pi^{\ast, \alpha}(\cdot \mid s_h^\alpha) , \pi_{t,h}^{\alpha}(\cdot \mid s_h^\alpha)\right) - \KL\left(\pi_{t+1,h}^{\alpha}(\cdot \mid s_h^\alpha) , \pi_{t,h}^{\alpha}(\cdot \mid s_h^\alpha)\right)\\
    & \ +\eta_t \left\| \pi_{t+1,h}^{\alpha}(\cdot \mid s_h^\alpha) - \pi_{t,h}^{\alpha}(\cdot \mid s_h^\alpha)\right\|_1\left\| \hat{g}_{t,h}(s_h^\alpha, \cdot)\right\|_\infty \\
    & \ + \eta_t \left\langle \pi^{\ast, \alpha}(\cdot \mid s_h^\alpha) -\pi_{t,h}^{\alpha}(\cdot \mid s_h^\alpha) , \hat{g}_{t,h}(s_h^\alpha, \cdot)+\lambda\log(\pi_{t,h}^{\alpha}(\cdot \mid s_h^\alpha)) \right\rangle \\
    \leq \ & \KL\left( \pi^{\ast, \alpha}(\cdot \mid s_h^\alpha) , \pi_{t,h}^{\alpha}(\cdot \mid s_h^\alpha)\right) - \KL\left(\pi_{t+1,h}^{\alpha}(\cdot \mid s_h^\alpha) , \pi_{t,h}^{\alpha}(\cdot \mid s_h^\alpha)\right) + \frac{\eta_t H}{\gamma_t}\sqrt{2\KL\left( \pi_{t+1,h}^{\alpha}(\cdot \mid s_h^\alpha),  \pi_{t,h}^{\alpha}(\cdot \mid s_h^\alpha)\right)}\\
     & \ + \eta_t \left\langle \pi^{\ast, \alpha}(\cdot \mid s_h^\alpha) -\pi_{t,h}^{\alpha}(\cdot \mid s_h^\alpha) , \hat{g}_{t,h}(s_h^\alpha, \cdot)+\lambda\log(\pi_{t,h}^{\alpha}(\cdot \mid s_h^\alpha)) \right\rangle  \\
     \leq \ & \KL\left( \pi^{\ast, \alpha}(\cdot \mid s_h^\alpha) , \pi_{t,h}^{\alpha}(\cdot \mid s_h^\alpha)\right) + \frac{\eta_t^2H^2}{2\gamma_t^2} + \eta_t \left\langle \pi^{\ast, \alpha}(\cdot \mid s_h^\alpha) -\pi_{t,h}^{\alpha}(\cdot \mid s_h^\alpha) , Q^{\lambda, \alpha}_{h}\left(s_h^\alpha,\cdot, \pi^\beta_t, \mu^{\mathcal{I}}\right) +\lambda\log(\pi_{t,h}^{\alpha}(\cdot \mid s_h^\alpha)) \right\rangle  \\
     & \ + \eta_t \left\langle \pi^{\ast, \alpha}(\cdot \mid s_h^\alpha) -\pi_{t,h}^{\alpha}(\cdot \mid s_h^\alpha) , \hat{g}_{t,h}(s_h^\alpha, \cdot)- Q^{\lambda, \alpha}_{h}\left(s_h^\alpha,\cdot, \pi^\beta_t, \mu^{\mathcal{I}}\right) \right\rangle \,.
\end{align*}

For the last term, 
\begin{align*}
    & \eta_t\left\langle \pi^{\ast, \alpha}(\cdot \mid s_h^\alpha) -\pi_{t,h}^{\alpha}(\cdot \mid s_h^\alpha) , \hat{g}_{t,h}(s_h^\alpha, \cdot)- Q^{\lambda, \alpha}_{h}\left(s_h^\alpha,\cdot, \pi^\beta_t, \mu^{\mathcal{I}}\right) \right\rangle \\
    = \ & \eta_t \left\langle  \pi^{\ast, \alpha}(\cdot \mid s_h^\alpha) -\pi_{t,h}^{\alpha}(\cdot \mid s_h^\alpha)  ,\frac{\mathbb{I}\{a_h = a\}\left(c_h + \hat{V}^{\lambda,\alpha}_{h+1}\left(s_{h+1}^\alpha, \pi_t^\alpha, \mu^\mathcal{I}\right)\right)}{\pi_{t,h}(s_h) + \gamma_t} - \frac{\mathbb{I}\{a_h = a\}\left(c_h + V^{\lambda,\alpha}_{h+1}\left(s_{h+1}^\alpha, \pi_t^\alpha, \mu^\mathcal{I}\right)\right)}{\pi_{t,h}(s_h) + \gamma_t}\right\rangle  \\
    + & \ \eta_t\left\langle  \pi^{\ast, \alpha}(\cdot \mid s_h^\alpha) -\pi_{t,h}^{\alpha}(\cdot \mid s_h^\alpha) ,\frac{\mathbb{I}\{a_h = a\}\left(c_h + V^{\lambda,\alpha}_{h+1}\left(s_{h+1}^\alpha, \pi_t^\alpha, \mu^\mathcal{I}\right)\right)}{\pi_{t,h}(s_h) + \gamma_t}- \left(c_h+ \EE_{\pi^\alpha} \left[V^{\lambda, \alpha}_{h + 1} \left(s_{h+1}^\alpha, \pi^\alpha, \mu^{\mathcal{I}}\right)\right]\right)\right\rangle  \\
    \leq \ & \frac{\eta_t}{\gamma_t} \left|\hat{V}^{\lambda,\alpha}_{h+1}\left(s_{h+1}^\alpha, \pi_t^\alpha, \mu^\mathcal{I}\right) - V^{\lambda,\alpha}_{h+1}\left(s_{h+1}^\alpha, \pi_t^\alpha, \mu^\mathcal{I}\right)\right| + \eta_t\epsilon_{t,h}  + \eta_t\zeta_{t,h} \,.
\end{align*}
where 
\begin{align*}
    \epsilon_{t,h} = \ & \left\langle \pi^{\ast, \alpha}(\cdot \mid s_h^\alpha)  ,\frac{\mathbb{I}\{a_h = a\}\left(c_h + V^{\lambda,\alpha}_{h+1}\left(s_{h+1}^\alpha, \pi_t^\alpha, \mu^\mathcal{I}\right)\right)}{\pi_{t,h}(s_h) + \gamma_t}- \left(c_h+ \EE_{\pi^\alpha} \left[V^{\lambda, \alpha}_{h + 1} \left(s_{h+1}^\alpha, \pi^\alpha, \mu^{\mathcal{I}}\right)\right]\right)\right\rangle \\
    \zeta_{t,h} = \ & \left\langle \pi_{t,h}^{\alpha}(\cdot \mid s_h^\alpha), \left(c_h+ \EE_{\pi^\alpha} \left[V^{\lambda, \alpha}_{h + 1} \left(s_{h+1}^\alpha, \pi^\alpha, \mu^{\mathcal{I}}\right)\right]\right) - \frac{\mathbb{I}\{a_h = a\}\left(c_h + V^{\lambda,\alpha}_{h+1}\left(s_{h+1}^\alpha, \pi_t^\alpha, \mu^\mathcal{I}\right)\right)}{\pi_{t,h}(s_h) + \gamma_t}\right\rangle \,.
\end{align*}
Choose $\beta_t=\frac{H+1}{H+t}$, by Lemma \ref{lem:value_convergence}, we have $\left(\hat{V}^{\lambda,\alpha}_{h+1}\left(s_{h+1}^\alpha, \pi_t^\alpha, \mu^\mathcal{I}\right) - V^{\lambda,\alpha}_{h+1}\left(s_{h+1}^\alpha, \pi_t^\alpha, \mu^\mathcal{I}\right)\right) \leq O \left(\sqrt{\frac{H^3\log(HSAT/\delta)}{t}}\right)$.

Define $\omega_t^k = \prod^t_{j=k+1} \left(1-\eta_j\lambda \right)$, and $\iota = \log(HSAT/\delta)$.
Similar to the proof of the exact gradient case, using the monotonicity condition and summing over $H$ and set $\eta_0\lambda = 1$, we therefore have
\begin{align*}
& \sum^H_{h=1} \EE_{\mu^{\ast,\alpha}_h}\left[\KL\left( \pi^{\ast, \alpha}(\cdot \mid s_h^\alpha) , \pi_{t+1,h}^{\alpha}(\cdot \mid s_h^\alpha)\right)   \right] \\
\leq \ & \left(1-\eta_t\lambda \right)\sum^H_{h=1} \EE_{\mu^{\ast,\alpha}_h}\left[\KL\left( \pi^{\ast, \alpha}(\cdot \mid s_h^\alpha) , \pi_{t,h}^{\alpha}(\cdot \mid s_h^\alpha)\right)\right]  + \sum^H_{h=1} \left(\epsilon_{t,h}+\zeta_{t,h}\right) + \frac{\eta_t^2H^3}{2\gamma_t^2} + \frac{\eta_t\sqrt{H^3\iota}}{\gamma_t\sqrt{t}} \\
= \ &  \sum^H_{h=1} \sum^t_{k=1} \eta_k \omega_t^k\left(\epsilon_{k,h}+\zeta_{k,h}\right) +  \sum^t_{k=1}\omega_t^k\left(\frac{\eta_k^2H^3}{2\gamma_k^2} + \frac{\eta_k\sqrt{H^3\iota}}{\gamma_t\sqrt{t}} \right)\\
\leq \ & O \left(A\sum^t_{k=1}\gamma_k\eta_k\omega_t^k + \sqrt{\log(A/\delta)\sum^t_{k=1}\eta_k^2(\omega_t^k)^2} + \max_{k \in [t]}\frac{\eta_k\omega_t^k\log(1/\delta)}{\gamma_k} +  \sum^t_{k=1}\omega_t^k\left(\frac{\eta_k^2H^3}{2\gamma_k^2} + \frac{\eta_k\sqrt{H^3\iota}}{\gamma_k\sqrt{t}} \right)\right) \,,
\end{align*}

where the last inequality is by Lemma \ref{lem:omd_est_error1} and Lemma \ref{lem:omd_est_error2}.

Take $\eta_k = \frac{1}{\lambda k^{c_\eta}}$, $\gamma_k = \frac{1}{k^{c_\gamma}}$, by Lemma \ref{lem:sequence1}, we have
\begin{align*}
     O\left(A\sum^t_{k=1}\gamma_k\eta_k\omega_t^k + \sqrt{\log(A/\delta)\sum^t_{k=1}\eta_k^2(\omega_t^k)^2}\right) = \ & O \left(\frac{A\log(t)}{ t^{c_\gamma}} + \sqrt{\frac{\log(A/\delta)}{t^{2c_\gamma}}}\right) \,, \\
    O \left(\sum^t_{k=1}\omega_t^k\left(\frac{\eta_k^2H^3}{2\gamma_k^2} + \frac{\eta_k\sqrt{H^3\iota}}{\sqrt{t}} \right)\right)
    =\ &  O \left(\frac{H^3}{t^{c_\eta - 2c_\gamma }} + \frac{\sqrt{H^3\iota}}{t^{1/2-c_\gamma}}\right) \,.
\end{align*}
By Lemma \ref{lem:sequence2}, we have
\begin{align*}
    O \left(\max_{k \in [t]}\frac{\eta_k\omega_t^k\log(1/\delta)}{\gamma_t} \right) = O \left(\frac{\log(1/\delta)}{t^{c_\eta - c_\gamma}}\right)
\end{align*}

Take $c_\eta = 3/4, c_\gamma = 1/4$, and integrating over $\alpha \in [0,1]$, we therefore have
\begin{align*}
    D\left(\pi^\mathcal{I}_{t+1}\right) \leq  O \left(\frac{A\log(t)}{ t^{3/4}} + \frac{\sqrt{\log(A/\delta)}}{t^{3/4}} + \frac{H^3}{t^{1/4}} + \frac{\sqrt{H^3\iota}}{t^{1/4}} + \frac{\log(1/\delta)}{t^{1/2}}\right)\,.
\end{align*}
\end{proof}

\begin{lem}\label{lem:value_convergence}
    For any $h,t,\alpha$, take $\beta_t=\frac{H+1}{H+t}$, and we have
    \begin{align*}
\left|\hat{V}^{\lambda,\alpha}_{h}\left(s_h^\alpha, \pi_t^\alpha, \mu^\mathcal{I}\right) - V^{\lambda,\alpha}_{h}\left(s_h^\alpha, \pi_t^\alpha, \mu^\mathcal{I}\right)\right| \leq O \left(\sqrt{\frac{H^3\log(HSAT/\delta)}{t}}\right)\,.
    \end{align*}
\end{lem}
\begin{proof}
Define an auxiliary sequence $\left\{\beta_t^i\right\}_{i=1}^t$ based on the learning rate.
$$
\beta_t=\frac{H+1}{H+t}, \quad \beta_t^0=\prod_{j=1}^t\left(1-\beta_j\right), \quad \beta_t^i=\beta_i \prod_{j=i+1}^t\left(1-\beta_j\right)\,.
$$

Let $k$ denote the number of times $s$ is visited at step $h$ at the beginning of episode $t$, and suppose s was previously visited at episodes $t^1, \ldots, t^k < t$ at the $h$-th step. 
Then by Lemma 11 of \cite{jin2022v}, we can write 
\begin{align*}
    \hat{V}^{\lambda,\alpha}_{t,h}\left(s_h^\alpha, \pi_t^\alpha, \mu^\mathcal{I}\right) 
    =\ & \sum^k_{i=1} \beta_k^i \left(c_{h}(s_h^\alpha, a_h^{\alpha, t^i}) + \hat{V}^{\lambda,\alpha}_{t,h+1}\left(s^{t^i}_{h+1}, \pi_{t^i}^\alpha, \mu^\mathcal{I}\right) \right)\,.
\end{align*}
We denote $\left(P_hV^{\lambda,\alpha}_{h+1}\left( \pi_{t^i}^\alpha, \mu^\mathcal{I}\right)\right)(s,a) = \EE_{\pi_{t^i}^\alpha, s^\prime\sim P_h(\cdot \mid s,a)} V^{\lambda,\alpha}_{h+1}\left(s^\prime, \pi_{t^i}^\alpha, \mu^\mathcal{I}\right)$. We also let $\hat{P}_h^t$ to denote the actual state tradition experienced under $\pi_{t}^\alpha, \mu^\mathcal{I}$ at episode $t$. 
By the Bellman equation, and using the fact that $\sum^t_{i=0} \beta_t^i = 1$, we have
\begin{align*}
    & V^{\lambda,\alpha}_{h}\left(s_h^\alpha, \pi_t^\alpha, \mu^\mathcal{I}\right) \\
    = \ & \beta_k^0 V^{\lambda,\alpha}_{h}\left(s_h^\alpha, \pi_t^\alpha, \mu^\mathcal{I}\right) + \sum^k_{i=1}\beta_k^i \EE_{\pi_{t^i}^\alpha}\left[ \left(c_{h}(s_h^\alpha, a^\alpha) + \left(P_hV^{\lambda,\alpha}_{h+1}\left( \pi_{t^i}^\alpha, \mu^\mathcal{I}\right)\right)(s_h^\alpha,a^\alpha)\right)\right] \\
    = \ &  \beta_k^0 V^{\lambda,\alpha}_{h}\left(s_h^\alpha, \pi_t^\alpha, \mu^\mathcal{I}\right) + \sum^k_{i=1}\beta_k^i \EE_{\pi_{t^i}^\alpha}\left[ \left(c_{h}(s_h^\alpha, a^\alpha) + \left(\left(P_h - \hat{P}_h^{t_i}\right)V^{\lambda,\alpha}_{h+1}\left( \pi_{t^i}^\alpha, \mu^\mathcal{I}\right)\right)(s_h^\alpha,a^\alpha)\right)\right] \\
    & \ - \beta_k^i V^{\lambda,\alpha}_{h+1}\left(s^{t^i}_{h+1}, \pi_{t^i}^\alpha, \mu^\mathcal{I}\right)
\end{align*}
Therefore, we have the following recursion. 
    \begin{align*}
        & \hat{V}^{\lambda,\alpha}_{t,h}\left(s_h^\alpha, \pi_t^\alpha, \mu^\mathcal{I}\right) - V^{\lambda,\alpha}_{h}\left(s_h^\alpha, \pi_t^\alpha, \mu^\mathcal{I}\right) \\
        \leq \ & \sum^k_{i=1}\EE_{\pi_{t^i}^\alpha}\left[ \beta_k^i \left(\left(P_h - \hat{P}_h^{t_i}\right)V^{\lambda,\alpha}_{h+1}\left( \pi_{t^i}^\alpha, \mu^\mathcal{I}\right)(s_h^\alpha,a^\alpha)\right)\right] + \beta_k^i\left(\hat{V}^{\lambda,\alpha}_{t^i,h+1}\left(s^{t^i}_{h+1}, \pi_{t^i}^\alpha, \mu^\mathcal{I}\right) - V^{\lambda,\alpha}_{h+1}\left(s^{t^i}_{h+1}, \pi_{t^i}^\alpha, \mu^\mathcal{I}\right)\right)\,.
    \end{align*}
    Applying Martingale concentration and using the property of the sequence $\beta_k^i$, we have
    \begin{align*}
        & \hat{V}^{\lambda,\alpha}_{t,h}\left(s_h^\alpha, \pi_t^\alpha, \mu^\mathcal{I}\right) - V^{\lambda,\alpha}_{h}\left(s_h^\alpha, \pi_t^\alpha, \mu^\mathcal{I}\right)\\
        \leq\ &  O \left(\sqrt{\frac{H^3\log(HSAT/\delta)}{t}}\right) + \sum^k_{i=1} \beta_k^i\left(\hat{V}^{\lambda,\alpha}_{t^i,h+1}\left(s^{t^i}_{h+1}, \pi_{t^i}^\alpha, \mu^\mathcal{I}\right) - V^{\lambda,\alpha}_{h+1}\left(s^{t^i}_{h+1}, \pi_{t^i}^\alpha, \mu^\mathcal{I}\right)\right)\,.
    \end{align*}

    We now upper bound this by induction. Suppose for $\tau = t^i$, $i = 1, \ldots, k$, $\forall h \in [H]$ 
    \begin{align*}
        \hat{V}^{\lambda,\alpha}_{\tau,h}\left(s_h^\alpha, \pi_t^\alpha, \mu^\mathcal{I}\right) - V^{\lambda,\alpha}_{h}\left(s_h^\alpha, \pi_t^\alpha, \mu^\mathcal{I}\right)\leq O \left(\sqrt{\frac{H^3\log(HSAT/\delta)}{t}}\right) \,.
    \end{align*}

    For $t=1$, the base case holds trivially. 
    Then we consider $\hat{V}^{\lambda,\alpha}_{t,h}\left(s_h^\alpha, \pi_t^\alpha, \mu^\mathcal{I}\right) - V^{\lambda,\alpha}_{h}\left(s_h^\alpha, \pi_t^\alpha, \mu^\mathcal{I}\right)$, and notice that $t$ is the next update after time step $k$,
    \begin{align*}
        & \sum^k_{i=1} \beta_k^i\left(\hat{V}^{\lambda,\alpha}_{t,h+1}\left(s^{t^i}_{h+1}, \pi_{t^i}^\alpha, \mu^\mathcal{I}\right) - V^{\lambda,\alpha}_{h+1}\left(s^{t^i}_{h+1}, \pi_{t^i}^\alpha, \mu^\mathcal{I}\right)\right) \\
        \leq \ & \sum^k_{i=1}  O \left(\frac{2}{t}\cdot \sqrt{\frac{H^3\log(HSAT/\delta)}{t}}\right) \leq  O \left(\sqrt{\frac{H^3\log(HSAT/\delta)}{t}}\right)\,.
    \end{align*}
    where the first inequality is by Lemma 4.1 of \cite{jin2018q}, $\max_{i\in[t]}\beta_t^i \leq \frac{2}{t}$, and the second inequality is due to $k \leq t$.

    Hence, we have the claimed result.

\end{proof}

\section{Auxiliary lemmas for gradient estimation}

Let $\mathbb{I}\{a_h = a\}$ be a one-hot vector where it is one at $a_h$. Suppose one run OMD with estimated gradient $\hat{g}_t$ to obtain $\pi_t$, where $\hat{g}_t(a) = \frac{\sigma_t \mathbb{I}\{a_t = a\}}{\pi_t(a) + \gamma_t}$, $\EE[\mathbb{I}\{a_t = a\}\mid \mathcal{F}_{t-1}] = \pi_t(a)$, $\EE[\sigma_t\mid \mathcal{F}_{t-1}] = g_t$, and $\mathcal{F}_{t-1}$ is the history generated up to time $t$. Here we let $\pi^\ast_t$ be the optimality policy, which is independent of $g_t, \hat{g}_t$.
\begin{lem}[Lemma 20 of \cite{bai2020near}]
\label{lem:omd_est_error1}
Let $c_1, c_2, \ldots, c_t$ be fixed positive numbers. Then with probability at least $1-\delta$,
$$
\sum_{i=1}^t c_i\left\langle \pi_i, g_i-\hat{g}_i\right\rangle=\mathcal{O}\left(A \sum_{i=1}^t \gamma_i c_i+\sqrt{\ln (A / \delta) \sum_{i=1}^t c_i^2}\right) \,.
$$

\end{lem}
\begin{lem}[Lemma 10 of \cite{cai2023uncoupled}, Lemma 18 of \cite{bai2020near}]\label{lem:omd_est_error2}
    Let $c_1, c_2, \ldots, c_t$ be fixed positive numbers. Then for any sequence $\pi_1^{\star}, \ldots, \pi_t^{\star} \in \Delta_A$ such that $\pi_i^{\star}$ is $\mathcal{F}_{i-1}$-measurable, with probability at least $1-\delta$,
$$
\sum_{i=1}^t c_i\left\langle \pi_i^{\star}, \widehat{g}_i-g_i\right\rangle=\mathcal{O}\left(\max _{i \leq t} \frac{c_i \ln (1 / \delta)}{\gamma_t}\right) \,.
$$
\end{lem}

\begin{lem}[Lemma 4 of \cite{cai2023uncoupled}]
\label{lem:sequence1}
    Let $0<h<1,0 \leq k \leq 2$, and let $t \geq\left(\frac{24}{1-h} \ln \frac{12}{1-h}\right)^{\frac{1}{1-h}}$. Then
$$
\sum_{i=1}^t\left(i^{-k} \prod_{j=i+1}^t\left(1-j^{-h}\right)\right) \leq 9 \ln (t) t^{-k+h} \,.
$$
\end{lem}
\begin{lem}[Lemma 4 of \cite{cai2023uncoupled}]
\label{lem:sequence2}
    Let $0<h<1,0 \leq k \leq 2$, and let $t \geq\left(\frac{24}{1-h} \ln \frac{12}{1-h}\right)^{\frac{1}{1-h}}$. Then
$$
\max _{1 \leq i \leq t}\left(i^{-k} \prod_{j=i+1}^t\left(1-j^{-h}\right)\right) \leq 4 t^{-k} \,.
$$
\end{lem}

\section{Proof of Theorem \ref{thm:linear_Q}}

\begin{proof}
    Similar to Theorem \ref{thm:exact_Q}, we have
    \begin{align*}
       &  \KL\left( \pi^{\ast, \alpha}(\cdot \mid s_h^\alpha) , \pi_{t+1,h}^{\alpha}(\cdot \mid s_h^\alpha)\right) \\
        \leq \ &  \KL\left( \pi^{\ast, \alpha}(\cdot \mid s_h^\alpha) , \pi_{t,h}^{\alpha}(\cdot \mid s_h^\alpha)\right) + \frac{\eta_t^2H^2}{2\gamma_t^2} + \eta_t \left\langle \pi^{\ast, \alpha}(\cdot \mid s_h^\alpha) -\pi_{t,h}^{\alpha}(\cdot \mid s_h^\alpha) , Q^{\lambda, \alpha}_{h}\left(s_h^\alpha,\cdot, \pi^\beta_t, \mu^{\mathcal{I}}\right) +\lambda\log(\pi_{t,h}^{\alpha}(\cdot \mid s_h^\alpha)) \right\rangle  \\
     & \ + \eta_t \left\langle \pi^{\ast, \alpha}(\cdot \mid s_h^\alpha) -\pi_{t,h}^{\alpha}(\cdot \mid s_h^\alpha) , \hat{g}_{t,h}(s_h^\alpha, \cdot)- Q^{\lambda, \alpha}_{h}\left(s_h^\alpha,\cdot, \pi^\beta_t, \mu^{\mathcal{I}}\right) \right\rangle \,.
    \end{align*}
    For the last term, 
    \begin{align*}
    & \eta_t\left\langle \pi^{\ast, \alpha}(\cdot \mid s_h^\alpha) -\pi_{t,h}^{\alpha}(\cdot \mid s_h^\alpha) , \hat{g}_{t,h}(s_h^\alpha, \cdot)- Q^{\lambda, \alpha}_{h}\left(s_h^\alpha,\cdot, \pi^\beta_t, \mu^{\mathcal{I}}\right) \right\rangle \\
    = \ & \eta_t \left\langle  \pi^{\ast, \alpha}(\cdot \mid s_h^\alpha) -\pi_{t,h}^{\alpha}(\cdot \mid s_h^\alpha)  ,\frac{\mathbb{I}\{a_h = a\} \hat{Q}^{\lambda, \alpha}_{t,h}\left(s^\alpha_h, a^\alpha_h, \pi^\alpha_t, \mu^{\mathcal{I}}\right) }{\pi_{t,h}(s_h) + \gamma_t} - \frac{\mathbb{I}\{a_h = a\} Q^{\lambda, \alpha}_{t,h}\left(s^\alpha_h, a^\alpha_h, \pi^\alpha_t, \mu^{\mathcal{I}}\right) }{\pi_{t,h}(s_h) + \gamma_t}\right\rangle  \\
    + & \ \eta_t\left\langle  \pi^{\ast, \alpha}(\cdot \mid s_h^\alpha) -\pi_{t,h}^{\alpha}(\cdot \mid s_h^\alpha) ,\frac{\mathbb{I}\{a_h = a\}Q^{\lambda, \alpha}_{t,h}\left(s^\alpha_h, a^\alpha_h, \pi^\alpha_t, \mu^{\mathcal{I}}\right)}{\pi_{t,h}(s_h) + \gamma_t}- \EE_{\pi_t^\alpha}\left[Q^{\lambda, \alpha}_{t,h}\left(s^\alpha_h, a^\alpha_h, \pi^\alpha_t, \mu^{\mathcal{I}}\right)\right]\right\rangle  \\
    \leq \ & \frac{\eta_t}{\gamma_t} \left|\hat{Q}^{\lambda, \alpha}_{h}\left(s_h^\alpha, a_h^\alpha, \pi^\beta_t, \mu^{\mathcal{I}}\right) - Q^{\lambda, \alpha}_{h}\left(s_h^\alpha,a_h^\alpha, \pi^\beta_t, \mu^{\mathcal{I}}\right)\right| + \eta_t\epsilon_{t,h}  + \eta_t\zeta_{t,h} \,.
\end{align*}
where 
\begin{align*}
    \epsilon_{t,h} = \ & \left\langle \pi^{\ast, \alpha}(\cdot \mid s_h^\alpha)  ,\frac{\mathbb{I}\{a_h = a\}Q^{\lambda, \alpha}_{t,h}\left(s^\alpha_h, a^\alpha_h, \pi^\alpha_t, \mu^{\mathcal{I}}\right)}{\pi_{t,h}(s_h) + \gamma_t}- \EE_{\pi_t^\alpha}\left[Q^{\lambda, \alpha}_{t,h}\left(s^\alpha_h, a^\alpha_h, \pi^\alpha_t, \mu^{\mathcal{I}}\right)\right]\right\rangle \\
    \zeta_{t,h} = \ & \left\langle \pi_{t,h}^{\alpha}(\cdot \mid s_h^\alpha), \EE_{\pi_t^\alpha}\left[Q^{\lambda, \alpha}_{t,h}\left(s^\alpha_h, a^\alpha_h, \pi^\alpha_t, \mu^{\mathcal{I}}\right)\right]- \frac{\mathbb{I}\{a_h = a\}Q^{\lambda, \alpha}_{t,h}\left(s^\alpha_h, a^\alpha_h, \pi^\alpha_t, \mu^{\mathcal{I}}\right)}{\pi_{t,h}(s_h) + \gamma_t}\right\rangle \,.
\end{align*}

    By the update rule, for any $a^\alpha$, we have
    \begin{align*}
         \hat{Q}^{\lambda, \alpha}_{h}\left(s_h^\alpha, a_h^\alpha, \pi^\beta_t, \mu^{\mathcal{I}}\right) - Q^{\lambda, \alpha}_{h}\left(s_h^\alpha,a_h^\alpha, \pi^\beta_t, \mu^{\mathcal{I}}\right)
         = \ & \phi(s_h^\alpha, a_h^\alpha, z_h^\alpha)^\top \left(\hat{\theta}_{t,h} - \theta_h^\ast\right)^\top V^{\lambda, \alpha}_{h + 1} \left(s_{h + 1}^\alpha, \pi^\alpha, \mu^{\mathcal{I}}\right)\\
         \leq \ & \Tilde{O}\left(dH \|\phi(s^\alpha_h,a^\alpha_h,z_h^\alpha)\|_{\Lambda_{t,h}^{-1}}\right)\,,
    \end{align*}
    where the inequality is by Lemma \ref{lem:8_7_agarwal}.

    Summing over $H$ and taking integral over $\alpha$, we therefore have
    \begin{align*}
    & \int_\mathcal{I}\sum^H_{h=1} \EE_{\mu^{\ast,\alpha}_h}\left[\KL\left( \pi^{\ast, \alpha}(\cdot \mid s_h^\alpha) , \pi_{t+1,h}^{\alpha}(\cdot \mid s_h^\alpha)\right)   \right] d\nu(\alpha)\\
    \leq \ & \left(1-\eta_t\lambda \right)\int_\mathcal{I}\sum^H_{h=1} \EE_{\mu^{\ast,\alpha}_h}\left[\KL\left( \pi^{\ast, \alpha}(\cdot \mid s_h^\alpha) , \pi_{t,h}^{\alpha}(\cdot \mid s_h^\alpha)\right)\right]d\nu(\alpha)  + \frac{\eta_t^2H^2}{2\gamma_t^2} + \eta_t\int_\mathcal{I}\sum^H_{h=1} \left(\epsilon_{t,h}+\zeta_{t,h}\right)d\nu(\alpha)\\
    & \ + \eta_t\int_\mathcal{I}\left(J^{\lambda, \alpha} (\pi^{\ast, \alpha}, \mu_t^\mathcal{I}) - J^{\lambda, \alpha} (\pi^{\alpha}, \mu_t^\mathcal{I})\right)d\nu(\alpha) + \eta_t\int_\mathcal{I}\sum^H_{h=1} \EE_{\mu^{\ast,\alpha}_h}\left[\Tilde{O}\left(dH \|\phi(s^\alpha_h,a^\alpha_h,z_h^\alpha)\|_{\Lambda_{t,h}^{-1}}\right)\right] d\nu(\alpha)\\
    \leq \ & \left(1-\eta_t\lambda \right)\int_\mathcal{I}\sum^H_{h=1} \EE_{\mu^{\ast,\alpha}_h}\left[\KL\left( \pi^{\ast, \alpha}(\cdot \mid s_h^\alpha) , \pi_{t,h}^{\alpha}(\cdot \mid s_h^\alpha)\right)\right] d\nu(\alpha)  +\frac{\eta_t^2H^2}{2\gamma_t^2}+ \eta_t\int_\mathcal{I}\sum^H_{h=1} \left(\epsilon_{t,h}+\zeta_{t,h}\right)d\nu(\alpha) \\
    & \ + \frac{\eta_t}{\gamma_t}\int_\mathcal{I}\sum^H_{h=1} \EE_{\mu^{\ast,\alpha}_h}\left[\Tilde{O}\left(dH \|\phi(s^\alpha_h,a^\alpha_h,  z_h^\alpha)\|_{\Lambda_{t,h}^{-1}}\right)\right] d\nu(\alpha)\,,
    \end{align*}
    where the last inequality  by proposition \ref{prop:monotone}, $\int_\mathcal{I}J^{\lambda, \alpha} (\pi^{\ast, \alpha}, \mu_t^\mathcal{I}) - J^{\lambda, \alpha} (\pi^{\alpha}, \mu_t^\mathcal{I}) d\nu(\alpha)\leq \int_\mathcal{I} J^{\lambda, \alpha} (\pi^{\ast, \alpha}, \mu_t^{\ast,\mathcal{I}}) - J^{\lambda, \alpha} (\pi^{\alpha}, \mu_t^{\ast,\mathcal{I}}) d\nu(\alpha)\leq 0$, by optimality.

    Set $\eta_0 \lambda = 1$, define $\omega_t^k = \prod^t_{j=k+1} \left(1-\eta_j\lambda \right)$, we have
    \begin{align*}
    D\left(\pi^\mathcal{I}_{t+1}\right)
    \leq  \ & (1-\eta_t\lambda) D\left(\pi^\mathcal{I}_{t}\right) + \frac{\eta_t^2H^2}{2\gamma_t^2}
    + \eta_t\int_\mathcal{I}\sum^H_{h=1} \EE_{\mu^{\ast,\alpha}_h}\left[\Tilde{O}\left(dH \|\phi(s^\alpha_{k,h},a^\alpha_{k,h}, z_{k,h}^\alpha)\|_{\Lambda_{t,h}^{-1}}\right)\right]d\nu(\alpha) \\
    & \ + \eta_t\int_\mathcal{I}\sum^H_{h=1} \left(\epsilon_{t,h}+\zeta_{t,h}\right)d\nu(\alpha)\\
    \leq \ &\int_\mathcal{I}\sum^t_{k=1}\omega_t^k\left(\frac{\eta_k^2H^3}{2\gamma_k^2} \right)d\nu(\alpha)+ \int_\mathcal{I}\sum^H_{h=1} \sum^t_{k=1}\eta_k \omega_t^k\EE_{\mu^{\ast,\alpha}_h}\left[\Tilde{O}\left(dH \|\phi(s^\alpha_{k,h},a^\alpha_{k,h}, z_{k,h}^\alpha)\|_{\Lambda_{t,h}^{-1}}\right)\right] d\nu(\alpha)\\
    & \ + \int_\mathcal{I}\sum^H_{h=1} \sum^t_{k=1} \eta_k \omega_t^k\left(\epsilon_{k,h}+\zeta_{k,h}\right)  d\nu(\alpha)\,.
    \end{align*}

    Similar to Theorem \ref{thm:approx_Q}, take $\eta_k = \frac{1}{\lambda k^{c_\eta}}$, $\gamma_k = \frac{1}{k^{c_\gamma}}$ we have
    \begin{align*}
        \int_\mathcal{I}\sum^H_{h=1} \sum^t_{k=1} \eta_k \omega_t^k\left(\epsilon_{k,h}+\zeta_{k,h}\right)  d\nu(\alpha)
        = \ &   \int_\mathcal{I}O\left(A\sum^t_{k=1}\gamma_k\eta_k\omega_t^k + \sqrt{\log(A/\delta)\sum^t_{k=1}\eta_k^2(\omega_t^k)^2}\right)   d\nu(\alpha)\\
        = \ & O \left(\frac{A\log(t)}{ t^{c_\gamma}} + \sqrt{\frac{\log(A/\delta)}{t^{2c_\gamma}}}\right)  \,,
    \end{align*}
    and 
    \begin{align*}
         \sum^t_{k=1}\omega_t^k\left(\frac{\eta_k^2H^3}{2\gamma_k^2} \right)
    =\ &  O \left(\frac{H^3}{t^{c_\eta - 2c_\gamma }} \right) \,.
    \end{align*}

    By Cauchy-Schwarz and Jensen's inequalities, we have
    \begin{align*}
        & \int_\mathcal{I}\sum^H_{h=1} \sum^t_{k=1}\eta_k \omega_t^k\EE_{\mu^{\ast,\alpha}_h}\left[dH \|\phi(s^\alpha_{k,h},a^\alpha_{k,h}, z_{k,h}^\alpha)\|_{\Lambda_{t,h}^{-1}}\right] d\nu(\alpha) \\
        = \ &  dH\int_\mathcal{I}\sum^H_{h=1} \sum^t_{k=1}\eta_k \omega_t^k\EE_{\mu^{\ast,\alpha}_h}\left[ \sqrt{\phi(s^\alpha_{k,h},a^\alpha_{k,h}, z_{k,h}^\alpha)^\top \Lambda_{t,h}^{-1}\phi(s^\alpha_{k,h},a^\alpha_{k,h}, z_{k,h}^\alpha)}\right] d\nu(\alpha) \\
        \leq \ &  dH\int_\mathcal{I}\sum^H_{h=1} \sum^t_{k=1}\eta_k \omega_t^k\sqrt{\EE_{\mu^{\ast,\alpha}_h}\left[ \phi(s^\alpha_{k,h},a^\alpha_{k,h}, z_{k,h}^\alpha)^\top \Lambda_{t,h}^{-1}\phi(s^\alpha_{k,h},a^\alpha_{k,h}, z_{k,h}^\alpha)\right]} d\nu(\alpha) \\
        \leq \ &  dH\int_\mathcal{I}\sum^H_{h=1} \sqrt{\sum^t_{k=1}\left(\eta_k \omega_t^k\right)^2}\sqrt{\sum^t_{k=1}\EE_{\mu^{\ast,\alpha}_h}\left[\phi(s^\alpha_{k,h},a^\alpha_{k,h}, z_{k,h}^\alpha)^\top \Lambda_{t,h}^{-1}\phi(s^\alpha_{k,h},a^\alpha_{k,h}, z_{k,h}^\alpha)\right]}  d\nu(\alpha)
    \end{align*}

    Then
    \begin{align*}
        & \sum^t_{k=1}\EE_{\mu^{\ast,\alpha}_h}\left[ \phi(s^\alpha_{k,h},a^\alpha_{k,h}, z_{k,h}^\alpha)^\top \Lambda_{t,h}^{-1}\phi(s^\alpha_{k,h},a^\alpha_{k,h}, z_{k,h}^\alpha)\right] \\
        \leq \ & 2\sum^t_{k=1}\log\left(1+\EE_{\mu^{\ast,\alpha}_h}\left[ \phi(s^\alpha_{k,h},a^\alpha_{k,h}, z_{k,h}^\alpha)^\top \Lambda_{t,h}^{-1}\phi(s^\alpha_{k,h},a^\alpha_{k,h}, z_{k,h}^\alpha)\right] \right) \,.
    \end{align*}

    By the definition of $\Lambda$, have
    \begin{align*}
        \det(\Lambda_{k,h}) 
        %= \ & \det(\Lambda_{k-1,h})\left(1 + \|\phi(s^\alpha_{k,h},a^\alpha_{k,h}, z_{k,h}^\alpha)\|_{\Lambda_{t,h}^{-1}}^2\right)\\
        = \ & \det(I)\prod^t_{k=1}\left(1 + \|\phi(s^\alpha_{k,h},a^\alpha_{k,h}, z_{k,h}^\alpha)\|_{\Lambda_{t,h}^{-1}}^2\right)\,.
    \end{align*}
    Therefore, 
    \begin{align*}\sum^t_{k=1}\log\left(1+\EE_{\mu^{\ast,\alpha}_h}\left[ \phi(s^\alpha_{k,h},a^\alpha_{k,h}, z_{k,h}^\alpha)^\top \Lambda_{t,h}^{-1}\phi(s^\alpha_{k,h},a^\alpha_{k,h}, z_{k,h}^\alpha)\right] \right)= \ & \log\left(\frac{\det(\Lambda_{k,h}) }{\det(I)}\right)\leq d\log\left(\frac{t+1}{d\tau}\right)\,.
    \end{align*}

    By Lemma \ref{lem:sequence1} with $c_\eta < 1$, we have
    \begin{align*}
         \sqrt{\sum^t_{k=1}\left(\eta_k \omega_t^k\right)^2} \leq \sqrt{\sum^t_{k=1}\eta_k^2 \omega_t^k} \leq \frac{\sqrt{\lambda}}{t^{c_\eta/2}}\,.
    \end{align*}

    Therefore, 
    \begin{align*}
         \int_\mathcal{I}\sum^H_{h=1} \sum^t_{k=1}\eta_k \omega_t^k\EE_{\mu^{\ast,\alpha}_h}\left[dH \|\phi(s^\alpha_{k,h},a^\alpha_{k,h}, z_{k,h}^\alpha)\|_{\Lambda_{t,h}^{-1}}\right] d\nu(\alpha) 
         \leq \frac{\sqrt{\lambda}d^2H^2\log(t+1/d\tau)}{t^{c_\eta/2}} \,,
    \end{align*}
    and 
    \begin{align*}
         D\left(\pi^\mathcal{I}_{t+1}\right)\leq O \left(\frac{H^3}{t^{c_\eta - 2c_\gamma }} + \frac{\sqrt{\lambda}d^2H^2\log(t+1/d\tau)}{t^{c_\eta/2 - c_\gamma}} + \frac{A\log(t)}{ t^{c_\gamma}} + \sqrt{\frac{\log(A/\delta)}{t^{2c_\gamma}}}\right) \,.
    \end{align*}

    Take $c_\eta = 4/5, c_\gamma = 1/5$, we have
    \begin{align*}
         D\left(\pi^\mathcal{I}_{t+1}\right)\leq O \left(\frac{H^3}{t^{3/5}} + \frac{\sqrt{\lambda}d^2H^2\log(t+1/d\tau)}{t^{1/5}} + \frac{A\log(t)}{ t^{1/5}} + \sqrt{\frac{\log(A/\delta)}{t^{1/5}}}\right) \,.
    \end{align*}
\end{proof}

\begin{lem}[Lemma 8.7 of \cite{agarwal2019reinforcement}]\label{lem:8_7_agarwal}
For any $h \in [H], t \in [T]$, $s, a, z$, and fix a $\delta \in (0,1)$, with a probability of at least $1 - \delta$, we have
    \begin{align*}
        \phi(s_h^\alpha, a_h^\alpha, z_h^\alpha)^\top \left(\hat{\theta}_{t,h} - \theta_h^\ast\right)^\top V^{\lambda, \alpha}_{h + 1} \left(s_{h + 1}^\alpha, \pi^\alpha, \mu^{\mathcal{I}}\right) \leq \Tilde{O}\left(dH \|\phi(s^\alpha_h,a^\alpha_h,z_h^\alpha)\|_{\Lambda_{t,h}^{-1}}\right)\,,
    \end{align*}
    where $\Tilde{O}$ hides the logarithmic dependency on $t, H$ and $1/\delta$.
\end{lem}

\newpage
\section*{NeurIPS Paper Checklist}

\begin{enumerate}

\item {\bf Claims}
    \item[] Question: Do the main claims made in the abstract and introduction accurately reflect the paper's contributions and scope?
    \item[] Answer:  \answerYes{} % Replace by \answerYes{}, \answerNo{}, or \answerNA{}.
    \item[] Justification: The abstract and introduction accurately and clearly state the main theoretical claims made and discuss the main contributions and scope.

\item {\bf Limitations}
    \item[] Question: Does the paper discuss the limitations of the work performed by the authors?
    \item[] Answer: \answerYes{} % Replace by \answerYes{}, \answerNo{}, or \answerNA{}.
    \item[] Justification: The paper clearly states the assumption made and gives examples of when the assumptions are satisfied. The paper also states clearly of the experimental settings, the computational resources needed, etc. 

\item {\bf Theory Assumptions and Proofs}
    \item[] Question: For each theoretical result, does the paper provide the full set of assumptions and a complete (and correct) proof?
    \item[] Answer: \answerYes{} % Replace by \answerYes{}, \answerNo{}, or \answerNA{}.
    \item[] Justification: All assumptions are clearly stated, and all proofs are included in the appendix. All theorems and lemmas are properly referenced. 

    \item {\bf Experimental Result Reproducibility}
    \item[] Question: Does the paper fully disclose all the information needed to reproduce the main experimental results of the paper to the extent that it affects the main claims and/or conclusions of the paper (regardless of whether the code and data are provided or not)?
    \item[] Answer: \answerYes{} % Replace by \answerYes{}, \answerNo{}, or \answerNA{}.
    \item[] Justification: All implementation details are given in the experiment section. Code will be released upon acceptance of this paper.

\item {\bf Open access to data and code}
    \item[] Question: Does the paper provide open access to the data and code, with sufficient instructions to faithfully reproduce the main experimental results, as described in supplemental material?
    \item[] Answer: \answerNo{} % Replace by \answerYes{}, \answerNo{}, or \answerNA{}.
    \item[] Justification: We will provide open access to the code and data upon acceptance of this paper. 

\item {\bf Experimental Setting/Details}
    \item[] Question: Does the paper specify all the training and test details (e.g., data splits, hyperparameters, how they were chosen, type of optimizer, etc.) necessary to understand the results?
    \item[] Answer: \answerYes{} % Replace by \answerYes{}, \answerNo{}, or \answerNA{}.
    \item[] Justification: The experimental settings are described in the experiment settings in detail to reproduce the paper. 

\item {\bf Experiment Statistical Significance}
    \item[] Question: Does the paper report error bars suitably and correctly defined or other appropriate information about the statistical significance of the experiments?
    \item[] Answer: \answerYes{} % Replace by \answerYes{}, \answerNo{}, or \answerNA{}.
    \item[] Justification: Yes, all figures for the experiments are shown with shading the one standard deviation error. 

\item {\bf Experiments Compute Resources}
    \item[] Question: For each experiment, does the paper provide sufficient information on the computer resources (type of compute workers, memory, time of execution) needed to reproduce the experiments?
    \item[] Answer: \answerYes{} % Replace by \answerYes{}, \answerNo{}, or \answerNA{}.
    \item[] Justification: All computational resources used have been stated in the experiment section. 
    
\item {\bf Code Of Ethics}
    \item[] Question: Does the research conducted in the paper conform, in every respect, with the NeurIPS Code of Ethics \url{https://neurips.cc/public/EthicsGuidelines}?
    \item[] Answer: \answerYes{} % Replace by \answerYes{}, \answerNo{}, or \answerNA{}.
    \item[] Justification: I have reviewed the code of ethics. 

\item {\bf Broader Impacts}
    \item[] Question: Does the paper discuss both potential positive societal impacts and negative societal impacts of the work performed?
    \item[] Answer: \answerNA{} % Replace by \answerYes{}, \answerNo{}, or \answerNA{}.
    \item[] Justification: The paper is foundational and theoretical research without particular application or deployment. 
    
\item {\bf Safeguards}
    \item[] Question: Does the paper describe safeguards that have been put in place for responsible release of data or models that have a high risk for misuse (e.g., pretrained language models, image generators, or scraped datasets)?
    \item[] Answer: \answerNA{} % Replace by \answerYes{}, \answerNo{}, or \answerNA{}.
    \item[] Justification: The paper is theoretical with no such risk.

\item {\bf Licenses for existing assets}
    \item[] Question: Are the creators or original owners of assets (e.g., code, data, models), used in the paper, properly credited and are the license and terms of use explicitly mentioned and properly respected?
    \item[] Answer: \answerYes{} % Replace by \answerYes{}, \answerNo{}, or \answerNA{}.
    \item[] Justification: The experimental environment used in this paper has been clearly referenced to the original paper. 

\item {\bf New Assets}
    \item[] Question: Are new assets introduced in the paper well documented and is the documentation provided alongside the assets?
    \item[] Answer: \answerYes{}{} % Replace by \answerYes{}, \answerNo{}, or \answerNA{}.
    \item[] Justification: The experiment environments have been described in detail in this paper. 

\item {\bf Crowdsourcing and Research with Human Subjects}
    \item[] Question: For crowdsourcing experiments and research with human subjects, does the paper include the full text of instructions given to participants and screenshots, if applicable, as well as details about compensation (if any)? 
    \item[] Answer: \answerNA{} % Replace by \answerYes{}, \answerNo{}, or \answerNA{}.
    \item[] Justification: The paper does not involve crowdsourcing nor research with human subjects.

\item {\bf Institutional Review Board (IRB) Approvals or Equivalent for Research with Human Subjects}
    \item[] Question: Does the paper describe potential risks incurred by study participants, whether such risks were disclosed to the subjects, and whether Institutional Review Board (IRB) approvals (or an equivalent approval/review based on the requirements of your country or institution) were obtained?
    \item[] Answer: \answerNA{} % Replace by \answerYes{}, \answerNo{}, or \answerNA{}.
    \item[] Justification: the paper does not involve crowdsourcing nor research with human subjects.

\end{enumerate}
\end{document}